\documentclass[journal,twoside,web]{ieeecolor}
\usepackage{generic}
\usepackage{cite}

\usepackage{hyperref}
\hypersetup{hidelinks=true}
\usepackage{textcomp}
\def\BibTeX{{\rm B\kern-.05em{\sc i\kern-.025em b}\kern-.08em
    T\kern-.1667em\lower.7ex\hbox{E}\kern-.125emX}}
\usepackage{graphicx}  

\usepackage{amsmath,amssymb,amsthm}
\usepackage{autobreak}
\allowdisplaybreaks
\usepackage{mathrsfs}
\usepackage{color}
\usepackage{booktabs}
\usepackage{array}
\usepackage{svg}
\usepackage{siunitx} 
\usepackage{enumerate} 
\usepackage{graphicx}
\usepackage{epstopdf}
\usepackage{caption}
\usepackage{multicol}  
\usepackage{multirow}  
\usepackage{flushend}
\usepackage{booktabs}
\usepackage[noend]{algpseudocode}
\usepackage{wrapfig}
\newtheorem{remark}{Remark} 
\newtheorem{lemma}{Lemma} 
\newtheorem{question}{Problem} 
\newtheorem{corollary}{Corollary} 
\newtheorem{theorem}{Theorem} 
\newtheorem{definition}{Definition} 
\newtheorem{proposition}{Proposition}
\newtheorem{assumption}{Assumption}
\usepackage[ruled]{algorithm2e}
\definecolor{qiblue}{rgb}{0, 0, 1}
\definecolor{qiblack}{rgb}{0, 0, 0}
\newcommand{\qs}[1]{\textcolor{qiblack}{#1}}

\begin{document}

\title{Distributed Model Predictive Control for Optimal Consensus of Constrained Heterogeneous Multi-agent Systems}
\author{Nan~Bai, ~Tao Liu, \IEEEmembership{Member, IEEE},~Qishao~Wang, \IEEEmembership{Senior Member, IEEE},~and Zhisheng~Duan, \IEEEmembership{Senior Member, IEEE}
	\thanks{Nan Bai is with the Department of Electronic and Computer Engineering, The Hong Kong University of Science and Technology, Hong Kong SAR (e-mail: eenanbai@ust.hk).
		Tao Liu is with the Department of Electrical and Electronic Engineering, The University of Hong Kong (HKU), Hong Kong SAR, and the HKU Shenzhen Institute of Research and Innovation, Shenzhen, China (e-mail: taoliu@eee.hku.hk).
		Qishao Wang is with the Department of Dynamics and Control, Beihang University, and Beijing Key Laboratory of Evolutionary Vertical Takeoff and Landing Aircraft Technology, Beijing, China (email: wangqishao@buaa.edu.cn).
		Zhisheng Duan is with the School of Advanced Manufacturing and Robotics, Peking University, Beijing, China (e-mail: duanzs@pku.edu.cn).}
	\thanks{The work was supported in part by the National Natural Science Foundation of China under grants T2121002, 62373025 and 12332004, in part by the Beijing Municipal Science \& Technology Commission and Administrative Commission of Zhongguancun Science Park under Grant Z241100005424001 and in part by the Fundamental Research Funds for the Central Universities.}
}

\maketitle

\begin{abstract}
This paper investigates the distributed optimal consensus control problem of constrained heterogeneous multi-agent systems within a model predictive control (MPC) scheme.
\qs{Both the control input sequence and
the dynamically feasible
consensus equilibrium are optimized simultaneously within the proposed MPC framework to improve consensus performance, yielding a coupled constrained optimization problem at each prediction time.
A distributed primal--dual algorithm is developed to solve the resulting optimization problem, and locally
verifiable conditions are derived to guarantee its convergence. Furthermore, sufficient terminal conditions are established
for the proposed MPC framework to guarantee the recursive feasibility and
asymptotic consensus of the closed-loop heterogeneous multi-agent
systems.}
Finally, numerical simulations
verify the effectiveness of the proposed approach.
\end{abstract}


\vspace{-0.5em}
\section{Introduction}
\qs{With advances in automation and networked systems, cooperative
	control of multi-agent systems has attracted considerable attention.
}
Through collaboration among individuals, the adaptability and robustness of the multi-agent system can be enhanced to address complex tasks.
\qs{Consensus control is a fundamental topic in cooperative
	multi-agent systems and has found broad applications such as robotic
	coordination, quadrotor formation, autonomous manufacturing, and
	smart grids~\cite{Ren2025,Oh2015,Savazzi2021,Wen20212}.
Conventional consensus control methods \cite{Olfati-Saber2007,longwang,Lin2017} primarily focus on driving agents to
reach an agreement, but do not specify how the
agreement should be achieved or which consensus state is preferable.
This motivates the optimal consensus control problem \cite{Shi2024,Xian2025,Liang2024,Cao2010,LiY2022,Zhang2026},
which incorporates a performance criterion into the controller
design. Depending on the specific formulation, the criterion may
account for transient disagreement, control effort, or the quality
of the resulting consensus equilibrium. Among the existing
formulations, linear-quadratic optimal consensus has been
extensively studied~\cite{Cao2010,LiY2022,Zhang2026}, as it admits a tractable
controller that balances transient consensus performance
and control effort while ensuring asymptotic consensus under
appropriate design.}

\qs{Notice that these methods \cite{Cao2010,LiY2022, Zhang2026} 
	generally assume
	a prescribed consensus equilibrium, typically the origin, thereby
	limiting the degrees of freedom available for performance
	optimization. To overcome this limitation, the consensus state is introduced as an additional decision variable in~\cite{Wang2021,Bai2023}, allowing the control input and the consensus state to be optimized simultaneously. Nevertheless, state and input constraints are not explicitly considered in these studies. Such constraints are ubiquitous in practical applications. For example, robots may operate within a bounded
	workspace~\cite{Gonzalez-Banos2002}, while actuator inputs must
	remain within admissible limits to ensure operational
	safety~\cite{DAversa2013}.
	Model predictive control (MPC) provides an effective framework for
	optimizing control performance while explicitly handling constraints, and has been successfully applied to
	industrial processes and aerial vehicles~\cite{Forbes2015,Slegers2006,Chen2021}.}
	For multi-agent systems, distributed MPC is a more promising way to take full advantage of the collaboration \cite{Negenborn2014}. Existing work on the consensus issue has been established for second-order integrator systems \cite{Zhu2018}, general linear systems \cite{Shi2018}, and nonlinear systems \cite{Gao2017}. In addition, as extensions of consensus problems, the formation and flocking tasks based on a given reference trajectory are addressed in \cite{Dai2017} and \cite{Lyu2021}, respectively.

\qs{While 
	most existing MPC methods for multi-agent systems adopt a
	prescribed equilibrium or reference trajectory, they do
	not exploit the potential performance benefits of optimizing the
	consensus equilibrium online.
	To relax this restriction, local
	artificial targets are introduced as optimization variables in the MPC formulation \cite{Hirche2020}, where agreement among local targets is promoted
	through soft disagreement penalties. Accordingly, the optimized local targets are not guaranteed to coincide at each MPC update and thus may not define a common
	consensus equilibrium.
	A more recent work~\cite{Bai2024TAC} moves a step further by
	imposing hard consensus constraints on local targets and jointly optimizing the resulting 
	consensus targets and control inputs. Nevertheless, it does not ensure that the
optimized target constitutes a feasible equilibrium for
every agent. Consequently, the corresponding
recursive feasibility and asymptotic consensus guarantees cannot be directly
extended to general heterogeneous linear
multi-agent systems.
Motivated by these limitations,
the dynamical feasibility of
the optimized equilibrium needs to be explicitly incorporated into the MPC formulation, while local equilibrium variables should also respect the consensus constraint. The resulting
optimization problem therefore contains a local coupling
between the control input and the consensus equilibrium, together
with a global consensus coupling. 
To solve optimization problems with coupling constraints,
the alternating direction method of multipliers (ADMM) \cite{Boyd2010} 
and related primal--dual
methods \cite{Falsone2020,Bastianello2022,Carli2020,Pan2023,Gao2017a} provide effective tools. However, a direct application of these methods may not fully exploit the particular structure of the considered problem, leading to auxiliary variable copies, additional communication, or heavy computation burden to obtain
the exact solution of local subproblems involving coupled decision
variables. Since the optimization problem must be solved
repeatedly within the MPC scheme, these considerations motivate a suitable distributed algorithm tailored to the underlying optimization problem.
}

\qs{Building on the aforementioned advances, this paper proposes an MPC framework equipped with a distributed solver to address the consensus control problem of constrained heterogeneous multi-agent systems. The main contributions of this paper are summarized as follows.}
\begin{enumerate}[(1)]
	\item \qs{A novel MPC-based optimal consensus formulation is proposed for heterogeneous multi-agent systems, where both the control input and the consensus equilibrium are optimized simultaneously. Different from existing simultaneous optimization formulations in \cite{Wang2021,Bai2023,Hirche2020,Bai2024TAC}, the consensus
		equilibrium is explicitly restricted to the admissible equilibrium set of each agent, ensuring that the optimized common equilibrium is dynamically feasible and maintainable under heterogeneous dynamics and local constraints.}
	\item 
\qs{	A distributed projected primal-dual algorithm is developed
	to solve the resulting MPC optimization problem at
	each prediction time. By exploiting the graph sparsity and
	the local projected gradient structure, the control sequence
	and consensus equilibrium are updated in parallel using only
	local and neighboring information. Locally verifiable
	conditions are further derived to guarantee convergence of the
	proposed algorithm.}
	\item 
	\qs{The closed-loop properties of the proposed MPC framework are
		rigorously analyzed. Suitable terminal ingredients are developed
		to guarantee the recursive feasibility and the asymptotic consensus of the heterogeneous multi-agent system.
		} 
\end{enumerate}

\textbf{Notation:} Throughout this paper, $\mathbb{R}^n$ denotes the set of $n$-dimensional real column vectors; $\mathbb{R}^{m\times n}$ denotes the set of $m\times n$ real matrices; col$\{x_1,x_2,\cdots,x_N\}\triangleq \left [x_1^{\top},x_2^{\top},\ldots,x_N^{\top}\right]^{\top}$ denotes the collection of vectors $x_i\in\mathbb{R}^{n_i}$, $i\in\{1,2,\ldots,N\}$;
$\mathbf{1}$ and $\mathbf{0}$ denote column vectors with all elements of 1 and 0 in proper dimensions, respectively;
for a matrix $\mathcal{A}\in\mathbb{R}^{m\times n}$, $\mathcal{A}_{ij}$ denotes the element located in the $i$th row and $j$th column of $\mathcal{A}$.
rand$(a,b)$ represents a random number with a uniform distribution on the interval $[a,b]$;
$\otimes$ denotes the Kronecker product of matrices; $\|*\|$ denotes the Euclidean norm of vectors;
$\langle*,*\rangle$ denotes the inner product of vectors;
for matrix ${A}\in\mathbb{R}^{n\times n}$, $A>\mathbf{0}~(A\geq \mathbf{0})$ means that $A$ is a positive (semi-)definite matrix;
 for matrices ${A}\in\mathbb{R}^{n\times n},~{B}\in\mathbb{R}^{n\times m},~\|{B}\|^2_{A}\triangleq {B}^{\top}{A}{B}$; $\mathscr{P}_{\mathbb{Q}}\{*\}$ means projecting $*$ onto space $\mathbb{Q}$; the Minkowski sum of two convex sets are described as $\mathcal{X}\oplus\mathcal{Y}=\{x+y\big| x\in\mathcal{X},~y\in\mathcal{Y}\}$; the Pontryagin set of two sets is $\mathcal{X}\ominus\mathcal{Y}=\{z\big| z+y\in \mathcal{X},\forall y \in \mathcal{Y}\}$.

\vspace{-6pt}
\section{Preliminaries and Problem Formulation}
\subsection{Problem Statement}
This paper focuses on the constrained heterogeneous multi-agent systems, and each agent has the following discrete-time linear dynamics
{\small\begin{align}
	x_i(k+1)&=A_ix_i(k)+B_iu_i(k),\notag\\
	x_i(k)&\in\mathcal{X}_i,~u_i(k)\in\mathcal{U}_i,\label{dyn}
\end{align}}%
where $A_i\in\mathbb{R}^{n\times n}$ is the system matrix and $B_i\in\mathbb{R}^{n\times m}$ is the control input matrix with full column rank for all $i\in\{1,2,\ldots,N\}$. Both matrices are constant and may be different for different agents but with the same dimension.
Assume $\mathcal{X}_i$ is a convex and closed subset of $\mathbb{R}^{n}$ and $\mathcal{U}_i$ is a convex and compact subset of $\mathbb{R}^{m}$ to restrict agent's behavior such as bounded system state and actuator saturation.

Considering the constrained dynamics (\ref{dyn}), the set of feasible equilibrium $\mathscr{E}_i$ of each agent is defined as follows.
\begin{definition}
	The set of feasible equilibrium for each agent $i$, $i=1,2,\ldots,N$, is defined as
{\small	\begin{align}
		\mathscr{E}_i=\big\{\big(z_i,u_i^e\big)~\Big| z_i\in\mathcal{X}_i,u_i^e\in\mathcal{U}_i,
		z_i=A_iz_i+B_iu_i^e\big\},\nonumber
	\end{align} }%
	where $z_i$ denotes the feasible equilibrium state and $u_i^e$ is the corresponding equilibrium control input.
\end{definition}
\begin{assumption}
	The feasible equilibrium set $\mathscr{E}_i$ is non-empty for each agent.
\end{assumption}

Agents are allowed to communicate with their neighbors through a communication network, which is characterized by an undirected and connected graph $\mathcal{G}=(\mathcal{V},\mathcal{E})$.
The adjacency matrix of $\mathcal{G}$ is defined as $\mathcal{A}\in\mathbb{R}^{N\times N}$, where
$\mathcal{A}_{ii}=0$ and $\mathcal{A}_{ij}=\mathcal{A}_{ji}=1$ if $(i,j)\in\mathcal{E}$, otherwise $\mathcal{A}_{ij}=\mathcal{A}_{ji}=0$. Let $\mathcal{N}_i=\{j\in\mathcal{V}|(i,j)\in\mathcal{E}\}$ denote the neighbor set of agent $i$.
The Laplacian matrix of $\mathcal{G}$ is denoted by $\mathcal{L}\in\mathbb{R}^{N\times N}$, where $\mathcal{L}_{ii}=\sum_{j=1}^{N}\mathcal{A}_{ij}$ and $\mathcal{L}_{ij}=-\mathcal{A}_{ij}$ for $i\neq j$. It follows from \cite{Olfati-Saber2007} that the Laplacian matrix $\mathcal{L}$ has a simple zero eigenvalue associated with eigenvector $\mathbf{1}$. Then, the consensus manifold \cite{Bai2023} can be defined as
{\small\begin{align}
	\mathcal{Z}=\{\check{Z}\triangleq\operatorname{col}\{\check{z}_1,\check{z}_2,\ldots,\check{z}_N\}\in\mathbb{R}^{Nn}\big|~(\mathcal{L}\otimes I_n)\check{Z}=\mathbf{0}\}.\label{mani}
\end{align}}%
This paper aims to enable agents to reach the optimized stable consensus state through collaboration, that is, each agent is desired to converge to $(z_i,u_i^e)\in\mathscr{E}_i$, where $z_i$ lies on the consensus manifold (\ref{mani}).
\vspace{-1em}
\subsection{MPC-based Problem Formulation}
\qs{In this paper, an MPC framework will be utilized to solve the optimal consensus control problem. Assume that $t_k$ is the real time of the $k$th prediction, $\mathcal{T}$ is the finite prediction horizon, and $\Delta t = t_{k+1}-t_k$ is the control period satisfying $1\leq \Delta t \leq\mathcal{T}$. 
	Let the vector $u_i(l|t_k)$ denote the predicted input of the $i$-th agent at time instant $l+t_k$ based on the $k$th prediction. Similarly, the corresponding predicted state is represented by $x_i(l|t_k)$, and the optimized equilibrium state and control input is denoted by $z_i(t_k)$ and $u_i^e(t_k)$, respectively.
	Then, define the cost function for each individual at the $k$th prediction as}
\vspace{-1em}
{\small\begin{align}
	J_i(\textbf{u}_i\qs{(t_k)},z_i\qs{(t_k)})=\sum_{l  = 0}^{\mathcal{T}-1}\Big(&\|x_i(l\qs{|t_k})-z_i\qs{(t_k)}\|^2_{Q_i}\nonumber\\+\|u_i(l\qs{|t_k})-u_i^e\qs{(t_k)}\|^2_{R_i}&\Big)
	+\|x_i(\mathcal{T}\qs{|t_k})-z_i\qs{(t_k)}\|^2_{P_i},\label{cost}
\end{align}}%
where $\mathbf{u}_i\qs{(t_k)}=\operatorname{col}\{u_i(0\qs{|t_k}),u_i(1\qs{|t_k}),\ldots,u_i(\mathcal{T}-1\qs{|t_k})\}$, and $Q_i$, $R_i$, and $P_i$ are positive definite matrices.
Since $B_i$ has the full column rank, the equilibrium control input can be computed by $u_i^e\qs{(t_k)}=D_iz_i\qs{(t_k)}$ with $D_i=(B_i^{\top}B_i)^{-1}B_i^{\top}(I_n-A_i)$. \qs{Therefore, the cost (\ref{cost}) can be regarded as a function of
	$\mathbf{u}_i(t_k)$ and $z_i(t_k)$, and a nonempty, closed, and convex admissible
	equilibrium state set $\tilde{\mathcal{Z}}_i$ is selected such that
	$
	\tilde{\mathcal{Z}}_i
	\subseteq
	\left\{
	z_i \,\middle|\,
	(I-A_i-B_iD_i)z_i=\mathbf{0},\;
	z_i\in\mathcal{X}_i,\;
	D_i z_i\in\mathcal{U}_i
	\right\}.
	$}
	
	\vspace{-0.25em}
\begin{remark}
	The assumption that $B_i$ has the full column rank is mild and extensively used \cite{Shi2018}. In practice, it implies that the system does not contain redundant inputs. Otherwise, one can choose those independent control channels for the system.
\end{remark}%

\vspace{-0.25em}
Then, the optimal consensus control problem \qs{within the MPC framework} can be formulated as follows.
\begin{question}[\qs{MPC-based optimal consensus control problem}]
	At time $t_k$ with the given $x_i(t_k)$, $i\in\{1,2,\ldots,N\}$, solve
	\begin{subequations}
{\small		\begin{align}
			\min _{\mathbf{U}(t_k),\mathbf{Z}(t_k)} \mathbf{J}(\mathbf{U}(t_k),\mathbf{Z}&(t_k))=\sum_{i=1}^N J_i(\mathbf{u}_i(t_k),z_i(t_k))\nonumber\\
			\text{s.t.}~ x_i(l+1|t_k)=&~A_ix_i(l|t_k)+B_iu_i(l|t_k),\label{dlabel}\\
			\tilde{\mathcal{L}}Z(t_k)=&~\mathbf{0},\label{LZ}\\
			z_i(t_k)\in&~\tilde{\mathcal{Z}}_i,\label{zlabel}\\
			x_i(0|t_k)=&~x_i(t_k),\label{xin}\\
			u_i(l|t_k)\in~&\mathcal{U}_i,~l \in[0,\mathcal{T}-1],\label{uconstraints}\\
			x_i(l|t_k)\in~&\mathcal{X}_i,	~l\in[0,\mathcal{T}],\label{xconstraints}\\
			x_i(\mathcal{T}|t_k)\in ~&\mathcal{X}_{i\mathcal{T}}\label{xconstraints2},
		\end{align}\label{problem1}}%
	\end{subequations}%
where $\tilde{\mathcal{L}}=\mathcal{L}\otimes I_n$, $\mathcal{X}_{i\mathcal{T}}$ is a nonempty closed and convex terminal set satisfying $\mathcal{X}_{i\mathcal{T}}\subseteq\mathcal{X}_i$, $\mathbf{U}(t_k)=\operatorname{col}\{\mathbf{u}_1(t_k)$, $\mathbf{u}_2(t_k)$, $\ldots,$ $\mathbf{u}_N(t_k)\}$, $\mathbf{u}_i(t_k)=\operatorname{col}\{$ $u_i(0|t_k)$, $u_i(1|t_k)$, $\ldots,$ $u_i(\mathcal{T}-1|t_k)\}$, and $\mathbf{Z}(t_k)=\operatorname{col}\{z_1(t_k)$, $z_2(t_k)$, $\ldots,$ $z_N(t_k)\}$.
\end{question}

\begin{assumption}
	The admissible equilibrium sets have a common relative interior
	point, i.e., $
		\bigcap_{i=1}^{N}\operatorname{ri}
		\bigl(\tilde{\mathcal Z}_i\bigr)
		\neq \emptyset.
$
\end{assumption}

\begin{remark}
\qs{The multi-agent system \eqref{dyn} considered in this paper consists of dynamically decoupled agents, and each agent is subject to local state and input constraints. Dynamic couplings and inter-agent coupled constraints, such as
	collision avoidance constraints, are not included in the present
	formulation and are left for future research.}
\end{remark}

\begin{remark}
\qs{The previous work \cite{Bai2024TAC}  also studies an
	MPC-based optimal consensus problem in which the control input and the consensus state are jointly optimized.
		However, its original formulation does not explicitly require
		the optimized consensus state to constitute an admissible
		equilibrium for every heterogeneous agent.}
		In contrast, this paper introduces a novel problem formulation (\ref{problem1}) tailored to heterogeneous systems. \qs{First,} the consensus state $z_i\qs{(t_k)}$ of each agent is restricted within its admissible equilibrium set $\tilde{\mathcal{Z}}_i$, as shown in (\ref{zlabel}). \qs{This guarantees that every optimized consensus state is associated with an admissible equilibrium pair $(z_i\qs{(t_k)},u_i^e\qs{(t_k)})$ that is dynamically feasible for the $i$th agent at each prediction time $\qs{t_k}$.} \qs{Second,} 
		the equilibrium-input deviation term $\|u_i(l\qs{|t_k})-u_i^e\qs{(t_k)}\|_{R_i}^2$ is explicitly incorporated into the cost function (\ref{cost}). \qs{This embeds the equilibrium information into the optimization objective and establishes a direct connection between transient performance optimization and steady-state feasibility.}
		These two features jointly integrate feasibility and stability considerations into the optimal consensus formulation and provide the foundation for the recursive feasibility and asymptotic consensus analysis developed in Section \uppercase\expandafter{\romannumeral 4}.
\end{remark}

\vspace{-0.5em}
\qs{Based on Problem~1, the MPC-based optimal consensus
	control scheme is summarized in Algorithm~\ref{alg3}.
	At each prediction time $t_k$, Problem~1 is solved using a
	distributed solver, which will be developed in Section~III. The first
	$\Delta t$ control moves of the resulting optimal input sequence $\mathbf{u}_i^\star(t_k)$
	are then applied, after which the optimization
	problem is updated using the newly measured states.}

\vspace{-0.5em}
\begin{algorithm}[!htbp]
	\caption{MPC-based Optimal Consensus Control Scheme}
	\label{alg3}
	\textbf{Initialize}{ initial value $x_i(0|t_k)=x_i(t_k)$
		at each update time $t_k$, $k=0,1,\ldots$}
	\begin{algorithmic}[1]
		\State{\qs{\textbf{solve} 
				Problem~1 to obtain  $\mathbf{u}_i^{\star}(t_k)=\operatorname{col}\{{u}_i^{\star}(0|t_k)$, ${u}_i^{\star}(1|t_k)$, $\ldots$, ${u}_i^{\star}(\mathcal{T}-1|t_k)\}$ and ${z}_i^{\star}(t_k)$};
			\State{\textbf{apply} ${u}_i^{\star}(\tau\big| t_k)$, $\tau \in [0,\Delta t-1]$ as the real-time input;}
			\State{\textbf{set} $t_{k+1}=t_k+\Delta t$ and $k=k+1$, and \textbf{return} to step~1.}		}
	\end{algorithmic}
\end{algorithm}

\vspace{-1.5em}
\section{Distributed Solver for the MPC Optimization Problem}
\qs{
	As described in Algorithm~1, the online implementation of the
	proposed MPC scheme requires Problem~1 to be solved repeatedly
	at each prediction time $t_k$.
	In contrast to the traditional MPC problem with predefined consensus target or reference (e.g.,
	\cite{Zhu2018,Shi2018,Gao2017}),
both the control input and the consensus target are viewed as decision variables. This provides greater flexibility in optimizing control performance. However, it also couples these decision variables in
the local objective \eqref{cost}.
To exploit this structure, we develop a projected primal-dual scheme based on the augmented Lagrangian to jointly optimize  $\textbf{u}_i\qs{(t_k)}$ and $z_i\qs{(t_k)}$.
Before presenting the algorithm, 
	the local constraints (\ref{uconstraints})-(\ref{xconstraints2}) are collected into a feasible control sequence
	set. For a fixed prediction time $t_k$, define
	$\tilde{\mathcal{U}}_i(t_k)\triangleq \big\{\mathbf{u}_i\in\mathbb{R}^{m\mathcal{T}}\big|~\forall l\in[0,\mathcal{T}-1]$, $u_i(\qs{l|t_k})\in\mathcal{U}_i$, $x_i(\qs{l|t_k})\in\mathcal{X}_i$, $x_i(\mathcal{T}\qs{|t_k})\in\mathcal{X}_{i\mathcal{T}}\big\}$,
	where $x_i(0|t_k)=x_i(t_k)$ and the predicted states are
	generated according to the linear dynamics~\eqref{dlabel}.
}%

\qs{In the following, we will design a distributed algorithm tailored for Problem 1 as the inner
	optimization solver of Algorithm~\ref{alg3}, followed by its convergence analysis.
	For simplicity, the time-indexed notation \((t_k)\) is omitted
	throughout the remainder of this section.}
	
	\vspace{-1em}
\subsection{Distributed Algorithm Design}
Consider the augmented Lagrangian of Problem~1 in the following form
{\small\begin{align}
	L_{\rho}(\mathbf{U},\mathbf{Z},\mathbf{\Lambda})=\mathbf{J}(\mathbf{U},\mathbf{Z})+\mathbf{\Lambda}^{\top}\tilde{\mathcal{L}}\mathbf{Z}+\frac{\rho}{2}\|\mathbf{Z}\|^2_{\tilde{\mathcal{L}}},\label{augmented}
\end{align}}%
where $\rho>0$ is the augmented Lagrangian parameter, and $\mathbf{\Lambda}=\operatorname{col}\{\lambda_1,\lambda_2,\ldots,\lambda_N\}\in\mathbb{R}^{Nn}$ is the Lagrangian multiplier with respect to the consensus constraint (\ref{LZ}).
\qs{Then, a distributed parallel algorithm tailored for Problem~1 is designed in the following Algorithm~\ref{alg2}.}

	\vspace{-0.25em}
\begin{remark}
	\qs{Since the consensus constraint \eqref{LZ} is the global coupling constraint among agents, while the remaining constraints in Problem 1 are locally defined, we consider the augmented Lagrangian function \eqref{augmented} with respect to the consensus constraint \eqref{LZ}. This treatment avoids introducing additional multipliers for local constraints. Meanwhile, the other local constraints are explicitly enforced by the projection operators in \eqref{A1} and \eqref{A2}.}
\end{remark}

	\vspace{-1em}
\begin{algorithm}[!hbtp]
	\caption{Distributed Parallel Optimal Consensus Algorithm}
	\label{alg2}
	\textbf{Initialize} {$q=0$, initial value $\mathbf{u}_i^0$, $z_i^0$ and $\lambda_i^0=\mathbf{0}$, individual cost function $J_{i}\left(\mathbf{u}_{i}, z_{i}\right)$, $\rho>0$, the auxiliary parameters $g_{u_i}>0$ and $g_{z_i}>0$, stopping tolerances $\varepsilon_i>0$ and $\vartheta_i>0$, 
		\qs{and the stopping flag $s_i^0=0$.}}
	\begin{algorithmic}[1]
		\State \textbf{repeat}
		\State	 {~~~~$\forall i\in\{1,2,\ldots,N\}$, \textbf{(parallel) update}
		{\small	\begin{subequations}
				\begin{align}
			\mathbf{u}_i^{q+1} =&\mathcal{P}_{\tilde{\mathcal{U}}_i}\Big\{	\mathbf{u}_i^{q}-g_{u_i}\nabla_{\mathbf{u}_i}J_i(\mathbf{u}_i^q,z_i^q)\Big\}\label{A1}\\
			z_i^{q+1}=&\mathcal{P}_{\tilde{\mathcal{Z}}_i}\bigg\{	{z}_i^{q}-g_{z_i}\Big[\nabla_{z_i}J_i(\mathbf{u}_i^q,z_i^q)+\sum_{j\in\mathcal{N}_i}(\lambda_i^q-\lambda_j^q)\nonumber\\
			&\qquad\qquad\qquad+\rho\sum_{j\in\mathcal{N}_i}(z_i^q-z_j^q) \Big]\bigg\}\label{A2}
		\end{align}
			\end{subequations}}%
	
			\State ~~~~\textbf{update} the Lagrange multiplier 
		{\small	\begin{align}
				\lambda_i^{q+1}=&\lambda_i^{q}+\rho z_i^{q+1},\label{disA3}
			\end{align}}%
			
			\State ~~~~\textbf{communicate} $z_i^{q+1}$ and $\lambda_i^{q+1}$  with neighbors.

	\State ~~~~\qs{\textbf{update} the local stopping flag 
		{\small	\begin{align*}
				s_i^{q+1}= \begin{cases} 1, & \text{if } | J_i(\mathbf u_i^{q+1},z_i^{q+1}) -J_i(\mathbf u_i^{q},z_i^{q})|\leq \varepsilon_i\\ & \text{and } \|\sum_{j\in\mathcal N_i}(z_i^{q+1}-z_j^{q+1})\|\leq \vartheta_i,\\ 0, & \text{otherwise}, \end{cases} 	
		\end{align*}}%
		and perform a distributed termination check}
	
	\State	~~~	\textbf{set} $q=q+1$
	}
	\State	\textbf{until} \qs {the distributed termination check confirms \(s_i^q=1\), $\forall i\in\{1,2,\ldots,N\}$
	}
\end{algorithmic}	

\end{algorithm}
	

\qs{In Algorithm~\ref{alg2}, the all-agent stopping condition can be
	checked in a distributed manner using a finite-time
	termination detection mechanism (e.g., \cite{R4}). If only a subset of agents satisfies the local stopping criterion,
	no agent stops independently; instead, all agents continue the
	iterations until the all-agent stopping condition is confirmed. In addition,
	the \textbf{parallelism} is reflected in the fact that the
	$\mathbf{u}_i$-subproblem~(\ref{A1}) and the
	$z_i$-subproblem~(\ref{A2}) can be solved in parallel, since both
	rely on the historical information from the previous iteration. Considering the system dynamics \eqref{dyn} of each agent, the detailed procedures of \eqref{A1}-\eqref{A2} could be further specified, which could be found in  Appendix \uppercase\expandafter{\romannumeral 1}.}

\vspace{-1em}
\subsection{Convergence Analysis}
Then, the convergence properties of Algorithm~\ref{alg2} are investigated in the following theorem.
\begin{theorem}
	\label{theorem 1}
	Suppose that Assumption~2 holds and Problem~1 is feasible at the
	considered prediction time.
	Then, the solution sequence $\{(\mathbf{U}^q,\mathbf{Z}^q)$, $q=1,2,\ldots\}$ with $\mathbf{U}^q=\operatorname{col}\{\mathbf{u}_1^q$, $\mathbf{u}_2^q$, $\ldots,$ $\mathbf{u}_N^q\}$ and $\mathbf{Z}^q=\operatorname{col}\{z_1^q$, $z_2^q$, $\ldots,$ $z_N^q\} $ generated by Algorithm~2 converges to the optimal solution $(\mathbf{U}^{\star},\mathbf{Z}^{\star})$, if the following conditions are satisfied
{\small\begin{align}
		\begin{aligned}
			G_{u_i}> \mathbf{0}, ~G_{u_i}^{-1}-L_{\delta_i}I_{m\mathcal{T}}> \mathbf{0},\\
		G_{z_i}> \mathbf{0},~ G_{z_i}^{-1}-(L_{\delta_i}+2\rho \mathcal{L}_{ii}) I_n> \mathbf{0}
		\end{aligned}\label{conver}
\end{align}}%
for each agent,	where
 $G_{u_i}=g_{u_i} I_{m\mathcal{T}}$, $G_{z_i}=g_{z_i}I_{n}$, and $L_{\delta_i}$ is the Lipschitz constant of the gradient of $J_i(\mathbf{u}_i,z_i)$, as specified in the proof.
\end{theorem}
\noindent
\textbf{Proof.} The proof is given in Appendix \uppercase\expandafter{\romannumeral 2}.\qed

\begin{remark}
	\qs{It can be found that the compact form of Algorithm~\ref{alg2} can be interpreted as a projected gradient primal-dual type method applied to the augmented Lagrangian \eqref{augmented}. Here a modified penalty term $\|\mathbf{Z}\|^2_{\tilde{\mathcal{L}}}$ for the consensus constraint is employed in the augmented Lagrangian function, differing from the typical form $\|\tilde{\mathcal{L}}\mathbf{Z}\|_2^2$.
		This variation
		is based on the fact that $\|\mathbf{Z}\|_{\tilde{\mathcal{L}}}={0}$ and $\|\tilde{\mathcal{L}}\mathbf{Z}\|={0}$ are equivalent \cite[Observation 7.1.6]{Horn2012}.
		Furthermore,
		the saddle point of (\ref{augmented}) is equivalent to the optimal solution of Problem~1, which can be proved through their respective definitions and using the Karush-Kuhn-Tucker conditions. 
		This modified form preserves the same consensus manifold and facilitates the distributed decomposition used in Algorithm~2. The convergence of the resulting iterations is established in Theorem~1.
		}
\end{remark}

\vspace{-0.25em}
\qs{The convergence conditions (\ref{conver}) show that the auxiliary
	parameters $G_{u_i}$ and $G_{z_i}$ can be set by each agent using
	only its local information.
	Together with the distributed execution procedure, the proposed
	Algorithm~\ref{alg2} can be implemented in a fully distributed
	fashion. Moreover, under the convergence conditions (\ref{conver}), the convergence rate of Algorithm~2 could be further estimated in the following.}

\vspace{-0.25em}
\begin{corollary}
	Let \(\mathbf{\Lambda}^\star\) be an optimal Lagrange multiplier associated with the optimal solution \((\mathbf{U}^\star,\mathbf{Z}^\star)\), and suppose that $\|\mathbf{\Lambda}^{\star}\|<\gamma$, where $\gamma$ is a bounded positive constant. Then,
	for the iterates $\{(\mathbf{U}^q,\mathbf{Z}^q,\mathbf{\Lambda}^q)$, $q=1$, $2$, $\ldots\}$ generated by Algorithm~2, define the weighted iterates
{\small	\begin{align}
		\hat{\mathbf{U}}=\frac{1}{t}\sum_{q=1}^{t}\mathbf{U}^q,\quad	\hat{\mathbf{Z}}=\frac{1}{t}\sum_{q=1}^{t}\mathbf{Z}^q,\quad	\hat{\mathbf{\Lambda}}=\frac{1}{t}\sum_{q=1}^{t}{\mathbf{\Lambda}}^q,
		\label{can}
	\end{align}}%
	and the convergence rate of Algorithm~2 can be estimated as
	{\small\begin{align}
		\begin{aligned}
			J(\hat{\mathbf{U}},\hat{\mathbf{Z}})-J({\mathbf{U}}^{\star},{\mathbf{Z}}^{\star})\leq \frac{1}{ t}\Gamma_0 ,\quad 
			\|\tilde{\mathcal{L}}  \hat{\mathbf{Z}}\|  \leq \frac{1}{ t\left(\gamma-\left\|\mathbf{\Lambda}^{\star}\right\|\right)}\Gamma_0,
		\end{aligned}\label{rate}
	\end{align}}%
	where
	$(\mathbf{U^{\star}},\mathbf{Z^{\star}})$ is the optimal solution, $\Gamma_0$ $=\frac{1}{2}$ $(\|{\mathbf{U}}^{\star}-$ ${\mathbf{U}}^{0}\|_{G_U^{-1}}^2$ $+\|{\mathbf{Z}}^{\star}-{\mathbf{Z}}^{0}\|_{G_Z^{-1}-\rho \tilde{\mathcal{L}}}^2+$ $\frac{\gamma^2}{\rho}$ $\lambda_{\max }\{\mathcal{L}\})$, 	$G_U=\operatorname{diag}\{$ $G_{u_1}$, $G_{u_2}$, $\ldots,$ $G_{u_N}\}$, $G_Z=\operatorname{diag}\{G_{z_1}$, $G_{z_2}$, $\ldots,$ $G_{z_N}\}$,
	and $\lambda_{\max }\{\mathcal{L}\}$ is the maximum eigenvalue of the Laplacian matrix~$\mathcal{L}$.
\end{corollary}
\noindent
\textbf{Proof.} The proof is given in Appendix \uppercase\expandafter{\romannumeral 3}.\qed

	\vspace{-0.5em}
\section{Closed-loop Analysis of the Proposed MPC Framework}
\qs{In the previous section,
	Algorithm~2 provides a distributed solver for Problem~1 at each prediction time $t_k$.
Then, the recursive feasibility of the proposed MPC framework and the asymptotic consensus of the closed-loop multi-agent system will be analyzed in this section.
}

\begin{theorem}\label{feasible}
Suppose that Problem~1 is feasible at the initial time $t_0$, \qs{and the optimal solution of Problem~1 is employed at each prediction time.} Then, the recursive feasibility of the proposed MPC framework is guaranteed and the agents achieve asymptotic consensus while satisfying constraints if weighting matrices $Q_i>\mathbf{0}$, $R_i>\mathbf{0}$, $P_i>\mathbf{0}$, and the following conditions are satisfied, $\forall i\in\{1,2,\ldots,N\}$,
	\\
	(1)	there exists a control gain matrix $K_i$ such that
	{\small\begin{align}
				\sum_{h=0}^{\Delta t -1}\|(A_i&+B_iK_i)^{h}\|^2_{Q_i+K_i^{\top}R_iK_i}	\label{Ki}\\& + \|(A_i+B_iK_i)^{\Delta t}\|_{P_i}^2
				-P_i\leq \mathbf{0};\nonumber
		\end{align}}%
	(2)	\qs{the terminal set $\mathcal{X}_{i\mathcal{T}}$ and the admissible equilibrium set 
	$
	\tilde{\mathcal{Z}}_i$ satisfy}
		\small{\begin{align}
			&\qs{K_i\mathcal{X}_{i\mathcal{T}}\oplus (D_i-K_i)\tilde{\mathcal{Z}}_i\subseteq {\mathcal{U}}_i,}\nonumber\\
			&
	\qs{	(A_i+B_iK_i)\mathcal X_{i\mathcal T}
				\oplus
				(I-(A_i+B_iK_i))\tilde{\mathcal{Z}}_i
				\subseteq
				\mathcal X_{i\mathcal T}.}
		\label{XTi}
		\end{align}}
\end{theorem}
\noindent
\textbf{Proof.} The proof is given in Appendix \uppercase\expandafter{\romannumeral 4}.\qed

\begin{remark}
	\qs{The proposed MPC framework differs from the standard MPC formulation 
	with a fixed equilibrium or a prespecified reference, since the consensus 
	equilibrium $z_i(t_k)$ is also optimized online and must constitute a 
	feasible equilibrium for heterogeneous agents. As a result, the 
	terminal conditions must be imposed uniformly with respect to the 
	admissible equilibrium set $\tilde{\mathcal{Z}}_i$. This motivates the 
	terminal conditions in Theorem~\ref{feasible}, which guarantee recursive 
	feasibility and asymptotic consensus of the proposed MPC framework.}
\qs{Moreover, the closed-loop guarantees in
	Theorem~\ref{feasible} are established under the assumption
	that Problem~1 is solved to optimality at each prediction time $t_k$.
	In practical implementation, Algorithm~2 is terminated after a
	finite number of iterations according to the prescribed
	stopping criteria, and the resulting solution is generally
	approximate. A rigorous characterization of the closed-loop
	properties under such finite optimization errors is beyond the
	scope of this paper.}
\end{remark}

\begin{remark}
	Due to the high flexibility in heterogeneous system dynamics, the terminal conditions (\ref{Ki})-(\ref{XTi}) are intentionally broad to accommodate various scenarios, making it challenging to prescribe a universal design method. For controllable systems, a possible construction procedure is described below. As for the condition (\ref{Ki}), one can set $P_i$ and $K_i$ based on the algebraic Riccati equation. Specifically, set the weighting matrix $P_i$ as the positive definite solution of the algebraic Riccati equation $
			P_i = A_i^\top P_i A_i - A_i^\top P_i B_i (R_i + B_i^\top P_i B_i)^{-1} B_i^\top P_i A_i + Q_i$,
		and design $K_i$ as the optimal LQR gain $
			K_i = - (R_i + B_i^\top P_i B_i)^{-1} B_i^\top P_i A_i$.
		Then, solve the Lyapunov equation
	$
				(A_i+B_iK_i)^{\top}{S}_i(A_i+B_iK_i)-{S}_i=-Q_i
$
		to obtain the positive definite matrix ${S}_i$. Subsequently, choose an admissible equilibrium set
		$
		\tilde{\mathcal{Z}}_i$, design $\mathcal{X}_{i\mathcal{T}}\subseteq\mathcal{S}_i(\beta_i)\ominus(-\tilde{\mathcal{Z}}_i)$ with $\mathcal{S}_i(\beta_i)=\{e_i|e_i^{\top}{S}_ie_i\leq\beta_i\}$, and set the scalar parameter $\beta_i > 0$ such that $\mathcal{X}_{i\mathcal{T}} \subseteq \mathcal{X}_i$. 
		Finally, verify whether the conditions in (\ref{XTi}) are satisfied. If not satisfied, change $\tilde{\mathcal{Z}}_i$ within the bounds of $\mathcal{X}_i$ and re-adjust $\beta_i$ accordingly.
	\qs{The above construction adapts the standard terminal design
		procedure for MPC with a prescribed equilibrium (e.g.,
		[23]) to the present framework where the terminal
		conditions must hold uniformly over the admissible equilibrium
		set $\tilde{\mathcal Z}_i$.}
	\qs{A possible extension of the terminal conditions to nonlinear systems would require
		analogous nonlinear terminal ingredients, which is left as an
		important direction for future research.}
\end{remark}

\begin{remark}
\qs{In the proposed MPC framework, the optimized consensus state $z_i^\star(t_k)$ is updated at each prediction time. This online re-optimization brings a trade-off between control performance and computational complexity. Compared with MPC methods with a pre-specified equilibrium, whose local decision variable is mainly the control input sequence $\mathbf u_i(t_k)\in\mathbb R^{m\mathcal T}$, the proposed formulation optimizes $(\mathbf u_i(t_k),z_i(t_k))\in\mathbb R^{m\mathcal T+n}$, resulting in additional inner iterations for the multi-agent system at each prediction time. 
To reduce the resulting computation burden in practical implementations, one can adopt a two-stage strategy. In the first stage, Algorithm~2 is executed to jointly optimize the control input sequence and the consensus equilibrium. Once the consensus errors become sufficiently small, e.g., at some time $t_{k'}$ satisfying $\sum_{j\in\mathcal{N}_i}\|x_i(t_{k'})-x_j(t_{k'})\|\leq\mu_i$ for a given threshold $\mu_i>0$, the optimal consensus state $z_i^\star(t_{k'})$ can be fixed as a reference consensus target for each agent. In the second stage, each agent can employ a  decentralized MPC strategy to track this fixed reference locally \cite[Chapter 2]{Rawlings2017}. This strategy reduces computational burden by avoiding unnecessary repeated optimization procedure in the final stage. }
\end{remark}

\begin{remark}
\qs{Integrator-type systems constitute a practically important
	subclass of linear systems, which are widely used in applications
	such as unmanned aerial vehicles~\cite{Lyu2021}, smart
	sensors~\cite{Zhang2015sensor}, and spacecraft
	rotation~\cite{Ren2008}.  The terminal conditions in
	Theorem~2 admit a simplified form for such systems.
	Specifically, consider an $r$th-order
	integrator-type agent described by}
	\\
	\vspace{-1em}
{\small	\qs{\begin{align}
				A_i=\begin{bmatrix}
						1  & \qs{\eta}_i & 0 & \cdots& 0  \\
						0  & 1 &\qs{\eta}_i &\cdots &0\\
						\vdots & \vdots & \ddots & \ddots &\vdots\\
						0 & 0 & 0  &1 &\qs{\eta}_i\\
						0 & 0 & 0  &0 &1
					\end{bmatrix}\otimes I_d,~~
				B_i=\begin{bmatrix}
						0\\0\\\vdots\\0 \\ \qs{\kappa}_i
					\end{bmatrix}\otimes I_d,\label{inter}          
			\end{align}}}%
	\qs{where $\eta_i>0$ and $\kappa_i>0$ denote the
	state-transition and input-channel coefficients, respectively.
	For this class of systems, the pair $(A_i,B_i)$ is controllable,
	and every equilibrium has the form
	$
	z_i=\operatorname{col}
	\{\alpha_i,\mathbf 0_{(r-1)d}\}$,
	$
	u_i^e=\mathbf 0_d$,
	where $\alpha_i\in\mathbb R^d$. Hence, provided that
	$\mathbf 0_d\in\mathcal U_i$, an admissible equilibrium set can
	be selected as
	$
	\tilde{\mathcal Z}_i
	\subseteq
	\left\{
	\operatorname{col}
	\{\alpha_i,\mathbf 0_{(r-1)d}\}
	\in\mathcal X_i
	\right\}
	$, and the first condition in
	\eqref{XTi} reduces to
	$
	K_i\bigl(
	\mathcal X_{i\mathcal T}
	\oplus(-\tilde{\mathcal Z}_i)
	\bigr)
	\subseteq
	\mathcal U_i$.
	The remaining terminal ingredients can then be constructed following the procedure in Remark~7 to satisfy the
	terminal conditions.}
\end{remark}

\vspace{-1 em}
\section{Simulation Examples}
In this section, the proposed method will be utilized in different scenarios to verify its effectiveness.

\vspace{-1em}
\subsection{Heterogeneous Linear Systems }
{First, consider a multi-agent system with five heterogeneous agents.}
{The communication topology is shown in Fig. \ref{topo}, and
	the dynamics of each agent are linear systems (\ref{dyn}) with}
{\small\begin{align*}
		{	A_i=	\begin{bmatrix}
				0& 1 &0 \\
				0& 0 &1 \\
				\Delta_i^1&\Delta_i^2  &\Delta_i^3
			\end{bmatrix},~~B_i=\begin{bmatrix}
				0\\0\\1
		\end{bmatrix}}
\end{align*}}%
{and $\Delta_i^j=\operatorname{rand}(-1,1)$, $\forall i\in\{1,2,3,4,5\}$, $\forall j\in\{1,2,3\}$. The state and input constraints of the
	agents are $\mathcal{X}_i=\big\{x_i=[x_{i\{1\}},x_{i\{2\}},x_{i\{3\}}]^{\top}\big| [-6,-6,-6]^{\top}\leq x_i\leq  [6,6,6]^{\top}\big\}$, $\mathcal{U}_i=\big\{u_i\big| -3\leq u_i\leq 3\big\}$, $i\in\{1,2,3,4,5\}$.}

\vspace{-0.5em}
\begin{figure}[!htbp]
	\centering
	\includegraphics[width=0.28 \linewidth]{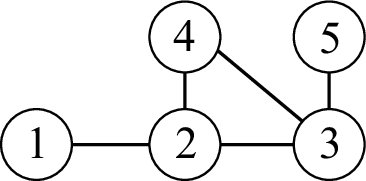}
	\caption{Communication topology of agents}
	\label{topo}
\end{figure}

\vspace{-0.5em}
\vspace{-0.5em}
The prediction horizon and control period are chosen as $\mathcal{T}=8$ and $\Delta t=2$, and the weighting matrices are set as $R_i=0.1$ and $Q_i=0.1I_3$.
\qs{Following the construction procedure discussed in Remark~7,
	for each agent, the terminal weighting matrix $P_i$ is obtained by the algebraic Riccati equation, and the feedback gain $K_i$ is computed accordingly. The matrix $S_i>\mathbf{0}$ is then obtained from
	the associated Lyapunov equation. Then, set the admissible equilibrium as 
	$\tilde{\mathcal Z_i}
	=
	\left\{
	z_i=\xi_i\mathbf1_3
	\mid
	-0.5\leq\xi_i\leq0.5
	\right\}$, and design the terminal set as
	\(
	\mathcal X_{i\mathcal T}
	=
	\left\{
	x_i
	\,\middle|\,
	x_i^\top S_i x_i\leq r_{i\mathcal T}^{\,2}
	\right\}\subseteq\mathcal S_i(\beta_i)
	\ominus
	\bigl(-\tilde{\mathcal Z}_i\bigr),
	\)
	where
	$\mathcal S_i(\beta_i)
	=\{e_i\mid e_i^\top S_i e_i\leq\beta_i\}$, where the parameters $r_{i\mathcal T}>0$ and
	$\beta_i>0$ are jointly selected
	such that the above inclusion and
	$\mathcal X_{i\mathcal T}\subseteq\mathcal X_i$ hold.
	The heterogeneous system parameters $\Delta_i=[\Delta_i^1,\Delta_i^2,\Delta_i^3]$ and selected terminal parameters
	$\beta_i$, $r_{i\mathcal T}$ are summarized in Table~I.
}%
The proposed Algorithm~1 is applied with the initial state $x_{i\{1\}}(t_0)=\operatorname{rand}(-6,6)$, $x_{i\{2\}}(t_0)=\operatorname{rand}(-6,6)$, $x_{i\{3\}}(t_0)=\operatorname{rand}(-6,6)$, \qs{and Algorithm~2 is used as the distributed solver with parameters} $g_{u_i}=g_{z_i}=\frac{1}{200}$, $\rho = 1$, $\mathbf{u}_i^0=\mathbf{0}$, ${z}_i^0(t_k)=x_i(t_k)$.

\begin{table}[!htbp]
	\centering
	\caption{\qs{Record of parameters $\Delta_i$, ${r}_{i\mathcal{T}}$ and $\beta_i$}}\label{table1}
	{
		\setlength{\tabcolsep}{1.5pt}
		\begin{tabular}{@{}cccccc@{}}
			\toprule
			{Agent} & {1} & {2} & {3} & {4} & {5} \\
			\midrule
			{ $\Delta_i$}
			&
			{\scriptsize $[0.4,0.2,0.3]$}
			&
			{\scriptsize  $[0.3,0.3,0.4]$}
			&
			{\scriptsize  $[0.5,0.4,0.3]$}
			&
			{\scriptsize  $[0.5,0.2,0.3]$}
			&
			{\scriptsize  $[0.3,0.4,0.4]$}
			\\
			{${r}_{i\mathcal{T}}$} & {$1.9227$} & {$1.9112$} & {$1.7141$} & {$1.9251$} & {$1.9108$} \\
			{$\beta_i$} & {$5.3591$} & {$5.3053$} & {$4.4463$} & {$5.3750$} & {$5.3063$} \\
			\bottomrule
	\end{tabular}}
\end{table}


\qs{Fig.~\ref{updatex} illustrates the system state trajectories and control inputs of agents. The state trajectories of all agents
	are retained in the first main panel to demonstrate their convergence
	to an optimized consensus equilibrium, while the inset
	shows the three state components of Agent~1 to make its transient
	behavior easily distinguished.}
Simulation results show that all agents quickly reach consensus, with both state variables and control inputs remaining within their prescribed constraints throughout the process. Notably, due to the heterogeneous dynamics, the control inputs converge to distinct nonzero values to maintain the consensus equilibrium.

\vspace{-0.5em}
\begin{figure}[!hbtp]
	\centering
	\includegraphics[width= 0.95\linewidth]{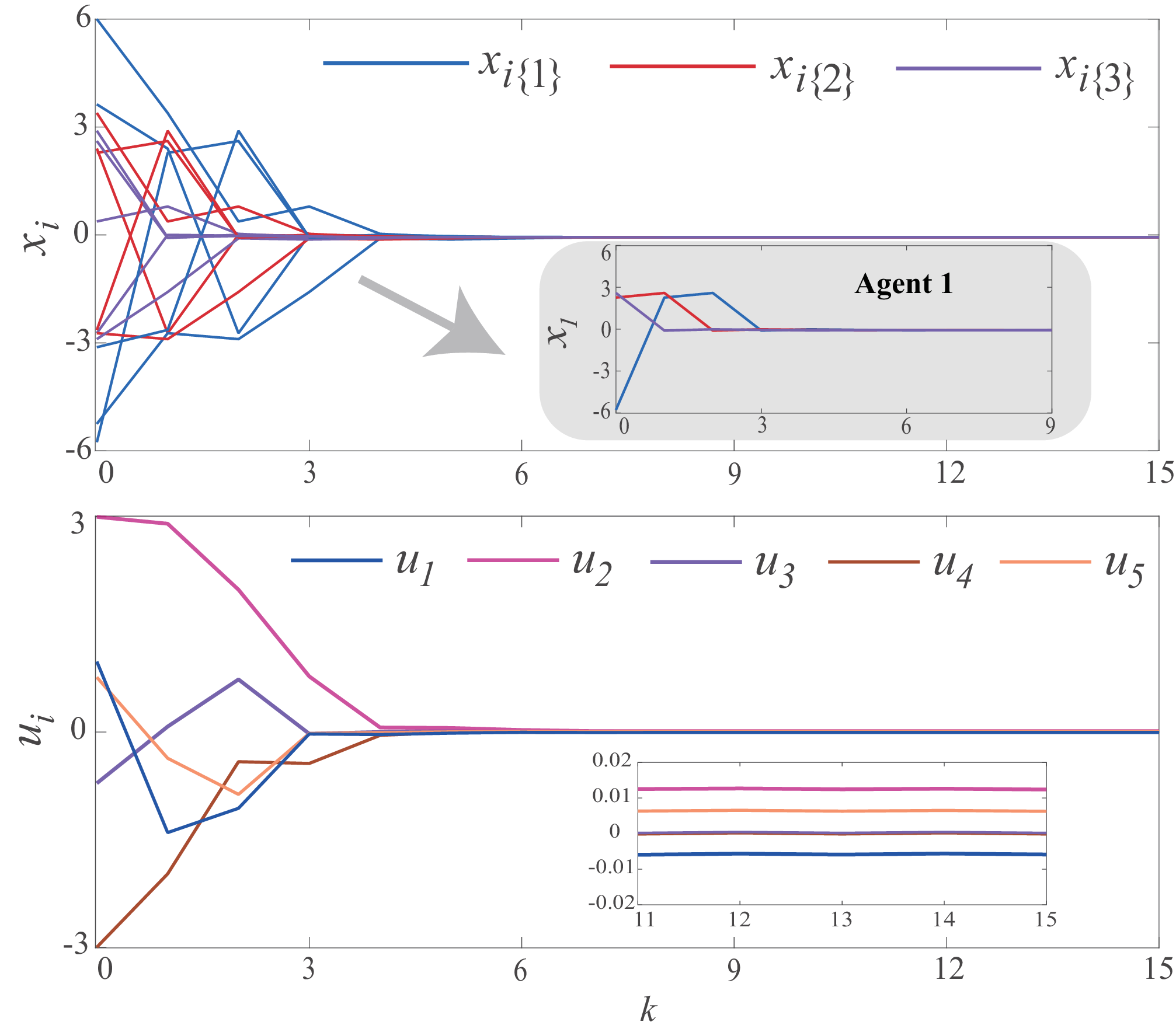}
	\caption{System states and control inputs of agents}
	\label{updatex}
\end{figure}

\vspace{-1.5em}
\subsection{Integrator Systems }
In this subsection, the proposed algorithm is applied to solve the optimal formation problem for a multi-robot system with five robots. Suppose the communication topology is the same as Fig.~\ref{topo}. \qs{Consider each robot modeled by the double-integrator dynamics
	\(
	m_i\ddot p_i(t)=u_i(t),
	\)
	where $m_i>0$, $p_i(t)\in\mathbb R^2$, and
	$u_i(t)\in\mathbb R^2$ denote the mass, position, and control
	force of robot $i$, respectively. Let
	$v_i(t)=\dot p_i(t)$ denote its velocity.
	Then, using a common sampling period $\delta$ and the forward-Euler
	discretization yields}
{\small
	\begin{align}
		\begin{bmatrix}
			p_i(k+1)\\
			v_i(k+1)
		\end{bmatrix}
		&=
		\begin{bmatrix}
			I_2 & \delta I_2\\
			\mathbf 0 & I_2
		\end{bmatrix}
		\begin{bmatrix}
			p_i(k)\\
			v_i(k)
		\end{bmatrix}
		+
		\begin{bmatrix}
			\mathbf 0\\
			\dfrac{\delta}{\qs{m_i}}I_2
		\end{bmatrix}
		u_i(k),
		\label{2order}
\end{align}}%
with $p_i=[p_{x_i},p_{y_i}]^{\top}$, $v_i=[v_{x_i},v_{y_i}]^{\top}$ and $u_i=[u_{x_i},u_{y_i}]^{\top}$.

To complete the formation task, 
let $p_{r_i}\in\mathbb R^2$ denote the relative position of the $i$th robot with respect to the formation center. 
Assume the formation pattern is set as $[p_{r_1}$, $p_{r_2}$, $p_{r_3}$, $p_{r_4}$, $p_{r_5}]=[(-2,2)^{\top},(-2,-2)^{\top},(2,2)^{\top},(2,-2)^{\top},(0,0)^{\top}]$, which forms a square with its center point.  The position,
velocity, and control input constraints are given by
$p_i\in\mathcal X_i=\{p_i\mid[-10,-10]^\top\leq
p_i\leq[10,10]^\top\}$,
$v_i\in\mathcal V_i=\{v_i\mid[-3,-3]^\top\leq
v_i\leq[3,3]^\top\}$, and
$u_i\in\mathcal U_i=\{u_i\mid[-3,-3]^\top\leq
u_i\leq[3,3]^\top\}$, respectively.
\qs{To represent the formation task within the consensus framework of Problem~1, define the shifted state
	$\tilde{x}_i=	\operatorname{col}\{\tilde p_i,v_i\}=\operatorname{col}\{p_i-p_{r_i},v_i\}$.
	It is verified that $\tilde{x}_i$ satisfies the
	same dynamics as in \eqref{2order}, and its equilibrium is defined as
	$z_i=\operatorname{col}\{c_i,\mathbf 0_2\}$, where $c_i$
	denotes the formation center. 
	Thus, consensus of $z_i$ yields a common formation center
	while preserving the prescribed relative positions $p_{r_i}$.
	The corresponding
	shifted state constraint set is
	\(
	\tilde{\mathcal X}_i
	=
	\left\{
	\operatorname{col}\{\tilde p_i,v_i\}
	\mid
	\tilde p_i+p_{r_i}\in\mathcal X_i,\;
	v_i\in\mathcal V_i
	\right\}.
	\)
	The cost function at each prediction time $t_k$ is then written as $J_i=\sum_{l=0}^{\mathcal{T}-1}\big(\|\tilde x_i(l|t_k)-z_i(t_k)\|_{Q_i}^2+\|u_i(l| t_k)\|^2_{R_i}\big)+\|\tilde x_i(\mathcal{T}| t_k)-z_i(t_k)\|_{P_i}^2$.
}

\qs{The common sampling period is selected as
	$\delta=0.5$. The nominal mass of each robot is
	$m_0=1$, while its actual mass is set as
	\(	
	m_i=m_0(1+\Delta m_i)\),
	\(\Delta m_i=\operatorname{rand}(-0.1,0.1).
	\)
	Then, set the prediction horizon as $\mathcal{T}=10$ and control period $\Delta t =1$. 
	The weighting matrices are set as $R_i=0.1 I_2$, $Q_i= I_4$.
	Following the terminal ingredient construction in Remark~7, the terminal weighting matrix
	$P_i$ and the feedback gain $K_i$ are obtained
	from the algebraic Riccati equation.
	The admissible
	equilibrium set is selected as $\tilde{\mathcal Z_i}
	=
	\left\{
	z_i=\operatorname{col}\{c_i,\mathbf{0}_2\}
	\mid
	-0.16\mathbf{1}_2\leq c_i\leq0.16\mathbf{1}_2
	\right\}$, and the terminal set is designed as 
	$
	\mathcal X_{i\mathcal T}
	=\{\tilde{x}_i\mid\tilde{x}_i^\top S_i \tilde{x}_i\leq r^2_{i\mathcal{T}}\}$.
	Moreover, with
	\(
	\mathcal S_i(\beta_i)
	=
	\left\{
	e_i
	\,\middle|\,
	e_i^\top S_i e_i\leq\beta_i
	\right\},
	\)
	the parameters $\beta_i>0$ and $r_{i\mathcal T}>0$ are jointly
	selected such that
	\(
	\mathcal X_{i\mathcal T}
	\subseteq
	\mathcal S_i(\beta_i)
	\ominus
	\bigl(-\tilde{\mathcal Z}_i\bigr)\) and $\mathcal X_{i\mathcal T}
	\subseteq
	\tilde{\mathcal X}_i$.
	The heterogeneous system parameters $m_i$ and selected terminal parameters $\beta_i$,
	$r_{i\mathcal T}$ are summarized in Table~\ref{table2}.
}
The proposed Algorithm~1 is utilized with the initial system state $p_{x_i}(t_0)$ $=\operatorname{rand}(-10,10)$, $p_{y_i}(t_0)$ $=\operatorname{rand}(-10,10)$, $v_{x_i}(t_0)$ $=\operatorname{rand}(-3,3)$, $v_{y_i}(t_0)$ $=\operatorname{rand}(-3,3)$, \qs{and Algorithm 2 is used as the distributed solver with parameters}
$g_{u_i}=g_{z_i}=\frac{1}{400}$, $\rho = 1$, $\mathbf{u}_i^0=\mathbf{0}$, $\qs{z_i^0(t_k)=
	\operatorname{col}\{p_i(t_k)-p_{r_i},\mathbf 0_2\}}$.
The formation process of robots is displayed in Fig. \ref{formation}, which shows that the formation task is completed.

\vspace{-0.25em}
\begin{table}[hbtp]
	\centering
	\caption{\qs{Record of parameters $m_i$, ${r}_{i\mathcal{T}}$ and $\beta_i$}}
	\label{table2}
	{
		\setlength{\tabcolsep}{6pt}
		\begin{tabular}{cccccc}
			\toprule
			\qs	{Robot} & \qs{1} &\qs {2} &\qs {3} &\qs {4} & \qs{5} \\
			\midrule
			{ $m_i$}
			&
			{$0.9259$}
			&
			{$0.9615$}
			&
			{$1.0000$}
			&
			{$1.0417$}
			&
			{$1.0870$}
			\\
			\qs{${r}_{i\mathcal{T}}$} & \qs{$1.6514$} & \qs{$1.6063$} & \qs{$1.5616$} &\qs {$1.5173$} & \qs{$1.4735$} \\
			\qs{$\beta_i$} & \qs{$4.3608$} &\qs {$4.1755$} &\qs {$3.9959$} & \qs{$3.8223$} & \qs{$3.6544$} \\
			\bottomrule
	\end{tabular}}
\end{table}

\vspace{-0.5em}
\begin{figure}[!htbp]
	\centering
	\includegraphics[width=0.95 \linewidth]{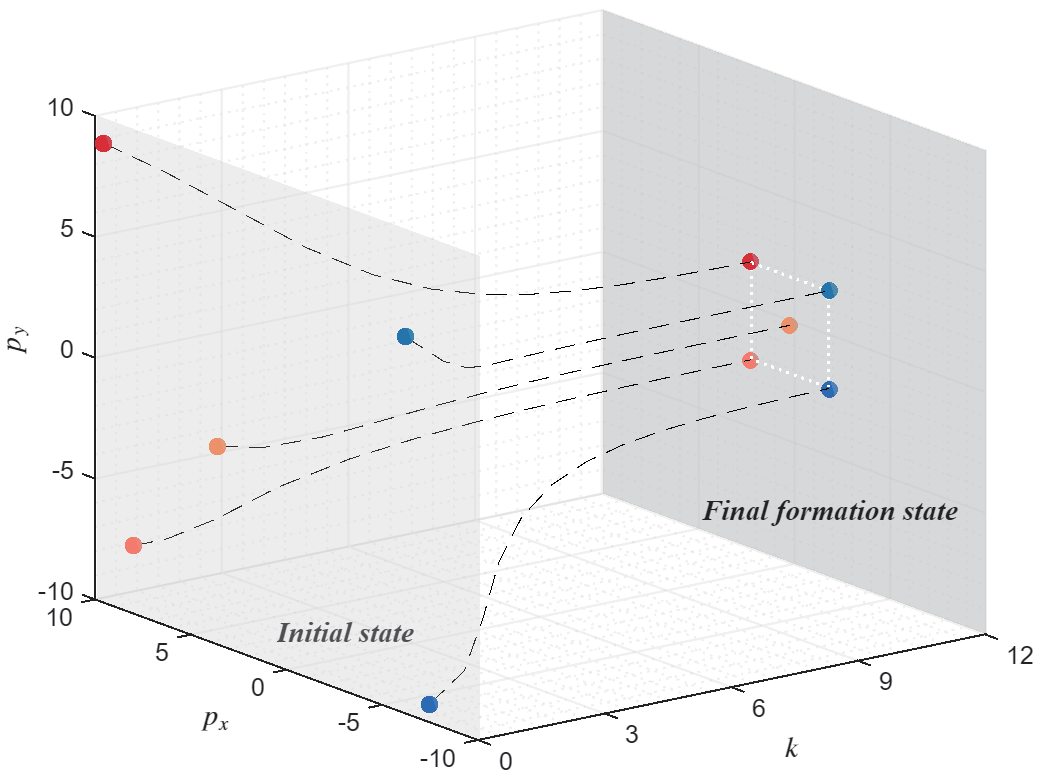}
	\caption{Formation process of robots}
	\label{formation}
\end{figure}

\qs{To further evaluate the closed-loop performance of large-scale multi-agent systems, 
	we consider rendezvous problems with
	$N=50$ and $N=100$ robots, where all desired relative
	positions are set to zero.} For comparison, the
system transformation-based consensus (STBC)
method~\cite{longwang}, the adaptive
feedback gain control (AFGC) method~\cite{Lin2017}, the sequential tracking MPC (STMPC)
method~\cite{Hirche2020}, \qs{and the optimization-based consensus
	MPC (OCMPC) method~\cite{Bai2024TAC} are implemented.
	Since some benchmark methods are developed for homogeneous
	agents, identical dynamics with $m_i=1$ in~\eqref{2order}
	are adopted for all methods.
	The comparison results are reported in
	Table~\ref{compare}, where
	the consensus step
	is defined as the first time instant $\mathscr T$ satisfying
	\(
	\sum_{i=1}^{N}\sum_{j=1}^{N}a_{ij}
	\| \tilde{x}_i(\mathscr T)- \tilde{x}_j(\mathscr T)\|
	\leq 10^{-4},
	\)
	and the performance cost is defined as
	\(
	\bar J
	=
	\sum_{t=0}^{\mathscr T-1}
	\sum_{i=1}^{N}
	[
	\sum_{j=1}^{N}a_{ij}
	(
	\| \tilde{x}_i(t)- \tilde{x}_j(t)\|_{Q_{i}}^2
	)+
	\|u_i(t)\|_{R_i}^2
	].
	\)
}

\vspace{-1em}
\qs{
	The comparison results in Table~\ref{compare} 
	demonstrate the advantages of the proposed approach. Since
	AFGC and STBC do not optimize the consensus equilibrium online, they have fewer degrees of
	freedom for improving the control performance, as reflected in the additional steps required to reach consensus and higher performance costs.
	STMPC introduces artificial targets to be optimized, but agreement among local targets is
	promoted through soft disagreement penalties. Hence, the
	optimized local targets may remain different during the MPC
	updates, potentially leading to slower convergence than that
	achieved by the proposed method and OCMPC, both of which
	enforce a hard consensus constraint.
	Although OCMPC achieves performance comparable to that of the
	proposed method in the considered examples, it is restricted
	to homogeneous agents and
	does not explicitly incorporate feasible equilibrium
	conditions. In contrast, the proposed method is applicable to more general constrained
	heterogeneous agents and ensures that the optimized common
	equilibrium is dynamically feasible for every agent, thereby
	providing a broader scope of applicability.}

\begin{table}[!hbtp]
	\centering
	\caption{\qs{Comparison results of different methods}}
	\label{compare}
	{
		\footnotesize
		\setlength{\tabcolsep}{4pt}
		\begin{tabular}{cccc}
			\toprule
			Method & Number of Robots & Consensus Step & Performance Cost\\
			\midrule
			Proposed & $50$  & $26$  & $3.6616\times10^4$\\
			OCMPC    & $50$  & $26$  & $3.6648\times10^4$\\
			STMPC    & $50$  & $68$  & $3.5440\times10^5$\\
			STBC     & $50$  & $87$  & $3.3821\times10^5$\\
			AFGC     & $50$  & $143$ & $5.1211\times10^5$\\
			\midrule
			Proposed & $100$ & $27$  & $4.6292\times10^4$\\
			OCMPC    & $100$ & $27$  & $4.6309\times10^4$\\
			STMPC    & $100$ & $87$  & $4.0371\times10^5$\\
			STBC     & $100$ & $127$ & $3.5506\times10^5$\\
			AFGC     & $100$ & $169$ & $5.7765\times10^5$\\
			\bottomrule
		\end{tabular}
	}
\end{table}%

\vspace{-0.5em}
\section{Conclusion}
\qs{In this paper, the optimal consensus control problem of constrained heterogeneous multi-agent systems is addressed within an MPC framework. In contrast to the classical MPC method with predefined equilibrium, 
both the control input and the consensus equilibrium are optimized in the MPC optimization problem. A distributed primal-dual algorithm is proposed to solve this optimization problem, where the control input sequence and consensus equilibrium are updated in parallel. Subsequently, the closed-loop property of the proposed MPC framework is analyzed, deriving sufficient conditions to ensure the recursive feasibility and asymptotic consensus of the multi-agent systems.} Finally, numerical simulations verify the validity of the proposed approach.

\appendices

\section{Proposition 1}
\begin{proposition}
	For the $i$th agent, the iteration process in (\ref{A1})-(\ref{A2}) can be implemented as 
	\begin{subequations}
		{	\begin{align}
				\mathbf{u}_i^{q+1}&=\mathscr{P}_{\tilde{\mathcal{U}}_i}\Big\{\operatorname{col}\big[\hat{\omega}_{u_i}^q(0),\hat{\omega}_{u_i}^q(1),\ldots,\hat{\omega}_{u_i}^q(\mathcal{T}-1)\big]\Big\},
				\label{disA1}\\
				{z}_i^{q+1}&=\mathscr{P}_{{\mathcal{Z}}_i}\big\{z_i^q-g_{z_i}F_{z_{ij}}^q\big\},\label{disA2}
		\end{align}}%
		\label{dis}
	\end{subequations}
	where 
	$\hat{\omega}_{u_i}^q(l)=u_i^q(l)-g_{u_i}F_{u_i}^q(l)$,
{	\begin{align}
			F_{u_i}^q(l)=&2\Big[\sum_{h=l+1}^{\mathcal{T}-1}(A_i^{h-l-1}B_i)^{\top} Q_i H_i^q(h)\nonumber\\
			&~+(A_i^{\mathcal{T}-l-1}B_i)^{\top} P_i H_i^q(\mathcal{T})+R_i (u_i^q(l)-D_iz_i^q)\Big],\nonumber\\
			H_i^q(l)=&A_i^lx_i(0)+\sum_{j=0}^{l-1}A_i^{l-1-j}B_iu_i^q(j)-z_i^q,\label{disFH}\\
			F_{z_{ij}}^q=-2&\Big[\sum_{l=0}^{\mathcal{T}-1}\big(Q_i(x_i^q(l)-z_i^q) +D_i^{\top}R_i(u_i^q(l)-D_iz_i^q)\big)\nonumber\\
			+&P_i(x_i^q(\mathcal{T})-z_i^q)\Big]+\sum_{j\in\mathcal{N}_i}\big[\rho(z_i^q-z_j^q)+(\lambda_i^q-\lambda_j^q)\big]\nonumber
	\end{align}}
	and $D_i=(B_i^{\top}B_i)^{-1}B_i^{\top}(I_n-A_i)$.
\end{proposition}
\noindent
\textbf{Proof. } This conclusion is derived by substituting the dynamic model \eqref{dyn} into the update processes \eqref{A1}-\eqref{A2}, so the proof is omitted for brevity.\qed

\section{Proof of Theorem 1}
Before the proof begins, some useful lemmas are given.

\begin{lemma}
	(\cite{POLYAK1987}) On a closed convex set $\mathbb{Q}$, the projection of $x\in\mathbb{R}^n$ possesses the following property:
	\begin{align*}
		\langle x-\mathscr{P}_{\mathbb{Q}}(x), y-\mathscr{P}_{\mathbb{Q}}(x)\rangle~\le~ 0, \quad \forall{x}\in\mathbb{R}^n,{y}\in \mathbb{Q}.
	\end{align*}
\end{lemma}

\begin{lemma}
	(\cite{Drori2015}) For any continuously differentiable convex function $f$ on $\mathbb{R}^p$ with $L_g$-Lipschitz continuous gradient $\nabla f$, one has $\forall{x,y,z\in\mathbb{R}^p}$,
	\begin{align*}
		f(x)\le ~f(y)+\langle\nabla f(z),x-y\rangle+~\frac{L_g}{2}\|x-z\|^2.
	\end{align*}
\end{lemma}

\noindent
\textbf{Then, the proof of Theorem 1 is presented as follows.}

\vspace{4pt}
\noindent
\textbf{Proof of Theorem 1.} 
Considering the equivalent form of (\ref{cost})
{\small\begin{align}
		J_i(\mathbf{u}_i,z_i)=&~(x_i(0)-z_i)^\top Q_i(x_i(0)-z_i)\nonumber\\
	
		+	&\sum_{k=1}^{\mathcal{T}-1}\Big[(A_i^kx_i(0)+\sum_{j= 0}^{k-1}A_i^{k-1-j}B_iu_i(j)-z_i)^{\top}Q_i\nonumber\\
		&\qquad\quad(A_i^kx_i(0)+\sum_{j= 0}^{k-1}A_i^{k-1-j}B_iu_i(j)-z_i)	\Big]\nonumber\\
	&	+\big(A_i^{\mathcal{T}}x_i(0)+\sum_{j=0}^{\mathcal{T}-1}A_i^{\mathcal{T}-1-j}B_iu_i(j)-z_i\big)^{\top}P_i\nonumber\\
	&~~~~\big(A_i^{\mathcal{T}}x_i(0)+\sum_{j=0}^{\mathcal{T}-1}A_i^{\mathcal{T}-1-j}B_iu_i(j)-z_i\big)\nonumber\\
		+	&\sum_{k=0}^{\mathcal{T}-1}(u_i(k)-D_iz_i)^{\top}R_i(u_i(k)-D_iz_i)
	\nonumber
\end{align}}%
with $D_i=(B_i^{\top}B_i)^{-1}B_i^{\top}(I_n-A_i)$,
then, the gradient of $J_i(\mathbf{u}_i,z_i)$ is
{\small \begin{align}
			\nabla_{z_i}J_i(\mathbf{u}_i,z_i)=&-2\Big[P_i\big(A_i^{\mathcal{T}}x_i(0)
		+\sum_{j=0}^{\mathcal{T}-1}A_i^{\mathcal{T}-1-j}B_iu_i(j)-z_i\big)\nonumber\\
		&+\sum_{k=1}^{\mathcal{T}-1}Q_i\big(A_i^{k}x_i(0)
		+\sum_{j=0}^{k-1}A_i^{k-1-j}B_iu_i(j)-z_i\big)\nonumber\\
		&+Q_i\big(x_i(0)-z_i)
		+\sum_{k=0}^{\mathcal{T}-1}D_i^{\top}R_i\big(u_i(k)-D_iz_i\big)
		\Big],\nonumber\\
	\nabla_{\mathbf{u}_i}J_i(\mathbf{u}_i,z_i)=&
	2\begin{bmatrix}
			\small(A_i^{\mathcal{T}-1}B_i)^{\top}P_i\big(A_i^{\mathcal{T}}x_i(0)\\
			+\sum_{j=0}^{\mathcal{T}-1}A_i^{\mathcal{T}-1-j}B_iu_i(j)-z_i\big) \\
			+\sum_{k=1}^{\mathcal{T}-1}\Big[ (A_i^{k-1}B_i)^{\top}Q_i\big(A_i^{k}x_i(0)\\
			+\sum_{j=0}^{k-1}A_i^{k-1-j}B_iu_i(j)-z_i\big)\Big] \\
			+R_i\big(u_i(0)-D_iz_i\big), \\
			(A_i^{\mathcal{T}-2}B_i)^{\top}P_i\big(A_i^{\mathcal{T}}x_i(0)\\
			+\sum_{j=0}^{\mathcal{T}-1}A_i^{\mathcal{T}-1-j}B_iu_i(j)-z_i\big) \\
			+\sum_{k=2}^{\mathcal{T}-1}\Big[ (A_i^{k-2}B_i)^{\top}Q_i\big(A_i^{k}x_i(0)\\
			+\sum_{j=0}^{k-1}A_i^{k-1-j}B_iu_i(j)-z_i\big)\Big] \\
			+R_i\big(u_i(1)-D_iz_i\big), \\
			\vdots\\
			B_i^{\top}P_i\big(A_i^{\mathcal{T}}x_i(0)\\
			+\sum_{j=0}^{{\mathcal{T}}-1}A_i^{{\mathcal{T}}-1-j}B_iu_i(j)-z_i\big) \\
			+R_i\big(u_i({\mathcal{T}}-1)-D_iz_i\big)
	\end{bmatrix}
	.\nonumber
\end{align}}%
Defining $\nabla_{[\mathbf{u}_i,z_i]} J_i\triangleq \big[\nabla_{\mathbf{u}_i}J_i^{\top}(\mathbf{u}_i,z_i),\nabla_{z_i}J_i^{\top}(\mathbf{u}_i,z_i)\big]^{\top}$ and {using the above results}, it is verified that
\begin{align}
	&\left\|\nabla_{[\mathbf{u}_i^p,z_i^p]} J_i(\mathbf{u}_i^p,z_i^p)-\nabla_{[\mathbf{u}_i^q,z_i^q]} J_i(\mathbf{u}_i^q,z_i^q)\right\|\nonumber\\
	\leq&L_{\delta_i} \left\|[ \mathbf{u}_i^{p^{\top}},z_i^{p^{\top}}]^{\top}-[ \mathbf{u}_i^{q^{\top}},z_i^{q^{\top}}]^{\top} \right\|,
\end{align}
where $L_{\delta_i}$ is the Lipschitz constant of the gradient of $J_i$. 

Recalling the local update rules in Algorithm~2, 
the compact form of them can be written as
\begin{subequations}
	\begin{align}
		\mathbf{U}^{q+1} =~& 
		\mathscr{P}_{\tilde{\mathcal{U}}}\big\{\mathbf{U}^q-G_U\nabla_\mathbf{U}L_{\rho}(\mathbf{U}^q,\mathbf{Z}^q,\mathbf{\Lambda}^q)\big\}, \label{CA1}\\ 
		\mathbf{Z}^{q+1}=~& \mathscr{P}_{\tilde{\mathcal{Z}}}\big\{\mathbf{Z}^q-G_Z\nabla_\mathbf{Z}L_{\rho}(\mathbf{U}^q,\mathbf{Z}^q,\mathbf{\Lambda}^q)\big\}, \label{CA2} \\
		\mathbf{\Lambda}^{q+1} = ~& \mathbf{\Lambda}^q+\rho \mathbf{Z}^{q+1},\label{CA3}
	\end{align}\label{Centralize}%
\end{subequations}%
where $G_U$ $=\operatorname{diag}$ $\{G_{u_1}$ $,G_{u_2},$ $\ldots,$ $G_{u_N}\}$, $G_Z=\operatorname{diag}\{G_{z_1},$ $G_{z_2},$ $\ldots,$ $G_{z_N}\}$, and
 $\tilde{\mathcal{U}}$ and $\tilde{\mathcal{Z}}$ correspond to $\mathbf{U}$ and $\mathbf{Z}$, respectively. Then, it follows from Lemma 1 that $\forall \mathbf{U}\in\tilde{\mathcal{U}}$, $\forall \mathbf{Z}\in\tilde{\mathcal{Z}}$,
\begin{align}
	(\mathbf{U}-\mathbf{U}^{q+1})^{\top}\Big[\mathbf{U}^{q+1}-\big(\mathbf{U}^{q}-G_U&\nabla_\mathbf{U}\mathbf{J}(\mathbf{U}^q,\mathbf{Z}^q)\big)\Big]\geq 0,\nonumber\\
	(\mathbf{Z}-\mathbf{Z}^{q+1})^{\top}\Big\{\mathbf{Z}^{q+1}-\big[\mathbf{Z}^{q}-G_Z(&\nabla_\mathbf{Z}\mathbf{J}(\mathbf{U}^q,\mathbf{Z}^q)\nonumber\\
	&+\tilde{\mathcal{L}}\mathbf{\Lambda}^q+\rho\tilde{\mathcal{L}}\mathbf{Z}^q)\big]\Big\}\geq 0,\nonumber
\end{align}
which are equivalent to $\forall \mathbf{U}\in\tilde{\mathcal{U}}$, $\forall \mathbf{Z}\in\tilde{\mathcal{Z}}$,
\begin{align}
	(\mathbf{U}&-\mathbf{U}^{q+1})^{\top}\Big[\nabla_\mathbf{U}\mathbf{J}(\mathbf{U}^q,\mathbf{Z}^q)+G_U^{-1}(\mathbf{U}^{q+1}-\mathbf{U}^{q})\Big]\geq 0\label{opt1}\\
	\text{and}~	&(\mathbf{Z}-\mathbf{Z}^{q+1})^{\top}\Big[\nabla_\mathbf{Z}\mathbf{J}(\mathbf{U}^q,\mathbf{Z}^q)+\tilde{\mathcal{L}}\mathbf{\Lambda}^q+\rho\tilde{\mathcal{L}}\mathbf{Z}^q\nonumber\\
	&\qquad\quad\qquad\qquad\qquad+G_Z^{-1}(\mathbf{Z}^{q+1}-\mathbf{Z}^{q}) \Big]\geq 0\label{opt2}
\end{align}
owing to $G_U> \mathbf{0}$ and $G_Z>\mathbf{0}$.
Notice that 
\begin{align}
	&(\mathbf{U}-\mathbf{U}^{q+1})^{\top}\nabla_\mathbf{U}\mathbf{J}(\mathbf{U}^q,\mathbf{Z}^q)+	(\mathbf{Z}-\mathbf{Z}^{q+1})^{\top}\nabla_\mathbf{Z}\mathbf{J}(\mathbf{U}^q,\mathbf{Z}^q)\nonumber\\
	=&(\mathbf{U}-\mathbf{U}^{q})^{\top}\nabla_\mathbf{U}\mathbf{J}(\mathbf{U}^q,\mathbf{Z}^q)+	(\mathbf{U}^q-\mathbf{U}^{q+1})^{\top}\nabla_\mathbf{U}\mathbf{J}(\mathbf{U}^q,\mathbf{Z}^q)\nonumber\\
	&+(\mathbf{Z}-\mathbf{Z}^{q})^{\top}\nabla_\mathbf{Z}\mathbf{J}(\mathbf{U}^q,\mathbf{Z}^q)+	(\mathbf{Z}^q-\mathbf{Z}^{q+1})^{\top}\nabla_\mathbf{Z}\mathbf{J}(\mathbf{U}^q,\mathbf{Z}^q)\nonumber\\
	\leq & \mathbf{J}(\mathbf{U},\mathbf{Z})-\mathbf{J}(\mathbf{U}^{q+1},\mathbf{Z}^{q+1})\nonumber\\
	&+\frac{1}{2}\|\mathbf{U}^q-\mathbf{U}^{q+1}\|^2_{L_{\delta}^U}+\frac{1}{2}\|\mathbf{Z}^q-\mathbf{Z}^{q+1}\|^2_{L_{\delta}^Z}\label{Lip1}
\end{align}
due to Lemma 2, where $L_{\delta}^U=\operatorname{diag}\{L_{\delta_1},L_{\delta_2},\ldots,L_{\delta_N}\}\otimes I_{m\mathcal{T}}$ and $L_{\delta}^Z=\operatorname{diag}\{L_{\delta_1},L_{\delta_2},\ldots,L_{\delta_N}\}\otimes I_{n}$.

Then, from the update rule (\ref{CA3}), it is shown that
\begin{align*}
	 \left(\mathbf{\Lambda}-\mathbf{\Lambda}^{q+1}\right)^{\top}\left[(1 / \rho) \tilde{\mathcal{L}}\left(\mathbf{\Lambda}^{q+1}-\mathbf{\Lambda}^{q}\right)-\tilde{\mathcal{L}} \mathbf{Z}^{q+1}\right]=0,
\end{align*}
and $\mathbf{\Lambda}^{q}=\mathbf{\Lambda}^{q+1}-\rho\mathbf{Z}^{q+1}$,
	which yields
\begin{align}
	&(\mathbf{Z}-\mathbf{Z}^{q+1})^{\top}\left(\tilde{\mathcal{L}} \mathbf{\Lambda}^{q}+\rho \tilde{\mathcal{L}} \mathbf{Z}^{q}\right)\nonumber \\
	=&\left(\mathbf{\Lambda}-\mathbf{\Lambda}^{q+1}\right)^{\top}\left[(1 / \rho) \tilde{\mathcal{L}}\left(\mathbf{\Lambda}^{q+1}-\mathbf{\Lambda}^{q}\right)-\tilde{\mathcal{L}} \mathbf{Z}^{q+1}\right]\nonumber \\
	&+\left(\mathbf{Z}-\mathbf{Z}^{q+1}\right)^{\top} \tilde{\mathcal{L}}\left[\mathbf{\Lambda}^{q+1}+\rho\left(\mathbf{Z}^{q}-\mathbf{Z}^{q+1}\right)\right]\nonumber \\
	=&-\frac{1}{2}\left\|\mathbf{\Lambda}-\mathbf{\Lambda}^{q+1}\right\|_{\frac{1}{\rho} \tilde{\mathcal{L}}}^{2}+\frac{1}{2}\left\|\mathbf{\Lambda}-\mathbf{\Lambda}^{q}\right\|_{\frac{1}{\rho} \tilde{\mathcal{L}}}^{2} \nonumber\\
	&-\frac{1}{2}\left\|\mathbf{\Lambda}^{q}-\mathbf{\Lambda}^{q+1}\right\|_{\frac{1}{\rho} \tilde{\mathcal{L}}}^{2}-\left(\mathbf{\Lambda}-\mathbf{\Lambda}^{q+1}\right)^{\top} \tilde{\mathcal{L}} \mathbf{Z} \label{Lip3}\\
	&+\left(\mathbf{Z}-\mathbf{Z}^{q+1}\right)^{\top} \tilde{\mathcal{L}} \mathbf{\Lambda}-\left(\mathbf{Z}-\mathbf{Z}^{q+1}\right)^{\top} \rho\tilde{\mathcal{L}}\left(\mathbf{Z}^{q+1}-\mathbf{Z}^{q}\right).\nonumber
\end{align}
Moreover, it is obtained that
{\begin{align}
		&(\mathbf{U}-\mathbf{U}^{q+1})^{\top}G_U^{-1}(\mathbf{U}^{q+1}-\mathbf{U}^q)\label{Lip2}\\
		=-\frac{1}{2}&\Big[ \|\mathbf{U}-\mathbf{U}^{q+1}\|^2_{G_U^{-1}}-\|\mathbf{U}-\mathbf{U}^{q}\|^2_{G_U^{-1}}+\|\mathbf{U}^{q+1}-\mathbf{U}^q\|^2_{G_U^{-1}}\Big],\nonumber\\
		&(\mathbf{Z}-\mathbf{Z}^{q+1})^{\top}G_Z^{-1}(\mathbf{Z}^{q+1}-\mathbf{Z}^q)\label{Lip4}\\
		=-\frac{1}{2}&\Big[ \|\mathbf{Z}-\mathbf{Z}^{q+1}\|^2_{G_Z^{-1}}-\|\mathbf{Z}-\mathbf{Z}^{q}\|^2_{G_Z^{-1}}+\|\mathbf{Z}^{q+1}-\mathbf{Z}^q\|^2_{G_Z^{-1}}\Big],\nonumber\\
	&\left(\mathbf{Z}-\mathbf{Z}^{q+1}\right)^{\top}\rho\tilde{\mathcal{L}}\left(\mathbf{Z}^{q+1}-\mathbf{Z}^{q}\right)\label{Lip5}\\
	=-\frac{1}{2}&\Big[ \|\mathbf{Z}-\mathbf{Z}^{q+1}\|^2_{\rho\tilde{\mathcal{L}}}-\|\mathbf{Z}-\mathbf{Z}^{q}\|^2_{\rho\tilde{\mathcal{L}}}+\|\mathbf{Z}^{q+1}-\mathbf{Z}^q\|^2_{\rho\tilde{\mathcal{L}}}\Big].\nonumber
\end{align}

Therefore, adding (\ref{opt1}) and (\ref{opt2}) and using (\ref{Lip1})-(\ref{Lip5}), it derives that:
\begin{align}
	\begin{aligned}
		0
		\leq & \mathbf{J}(\mathbf{U}, \mathbf{Z})-\mathbf{J}\left(\mathbf{U}^{q+1}, \mathbf{Z}^{q+1}\right)-\left(\mathbf{\Lambda}-\mathbf{\Lambda}^{q+1}\right)^ {\top} \tilde{\mathcal{L}}\mathbf{Z} \\
		&+\left(\mathbf{Z}-\mathbf{Z}^{q+1}\right)^ {\top} \tilde{\mathcal{L}} \mathbf{\Lambda}+\frac{1}{2}\left\|\mathbf{\Theta}-\mathbf{\Theta}^{q}\right\|_{\mathcal{G}_{1}}^{2} \\
		&-\frac{1}{2}\left\|\mathbf{\Theta}-\mathbf{\Theta}^{q+1}\right\|_{\mathcal{G}_{1}}^{2}-\frac{1}{2}\left\|\mathbf{\Theta}^{q}-\mathbf{\Theta}^{q+1}\right\|_{\mathcal{G}_{2}}^{2},
	\end{aligned}\label{M1}
\end{align}
where $\mathbf\Theta
=
\operatorname{col}
\{\mathbf U,\mathbf Z,\mathbf\Lambda\}$, $\mathcal{G}_1=\text{diag}\{{G}_U^{-1},{G}_Z^{-1}-\rho\tilde{\mathcal{L}},\frac{1}{\rho}\tilde{\mathcal{L}}\}$, and $\mathcal{G}_2=\text{diag}\{{G}_U^{-1}-L_{\delta}^U,G_Z^{-1}-\rho\tilde{\mathcal{L}}-L_{\delta}^Z,\frac{1}{\rho}\tilde{\mathcal{L}}\}$.

By Assumption~2 and the feasibility of Problem~1, the
relative interior constraint qualification holds. Hence, there
exists an optimal Lagrange multiplier
$\mathbf\Lambda^\star$ associated with the consensus constraint \eqref{LZ},
such that
$(\mathbf U^\star,\mathbf Z^\star,\mathbf\Lambda^\star)$
is a saddle point of the Lagrangian
\(
L(\mathbf U,\mathbf Z,\mathbf\Lambda)
=
\mathbf{J}(\mathbf U,\mathbf Z)
+
\mathbf\Lambda^\top\tilde{\mathcal L}\mathbf Z.
\)
Therefore, we have
	\begin{align*}
		&\mathbf{J}(\mathbf{U}^{q+1},\mathbf{Z}^{q+1})-\mathbf{J}(\mathbf{U}^{\star},\mathbf{Z}^{\star})+(\mathbf{\Lambda}^{\star}-\mathbf{\Lambda}^{q+1})^{\top}\tilde{\mathcal{L}}\mathbf{Z}^{\star}\\
		&-(\mathbf{Z}^{\star}-\mathbf{Z}^{q+1})^{\top}\tilde{\mathcal{L}}\mathbf{\Lambda}^{\star}\\
		=&\mathbf{J}(\mathbf{U}^{q+1},\mathbf{Z}^{q+1})-\mathbf{J}(\mathbf{U}^{\star},\mathbf{Z}^{\star})-(\mathbf{Z}^{\star}-\mathbf{Z}^{q+1})^{\top}\tilde{\mathcal{L}}\mathbf{\Lambda}^{\star}\\
		=&L(\mathbf{U}^{q+1},\mathbf{Z}^{q+1},\mathbf{\Lambda}^{\star})-L(\mathbf{U}^{\star},\mathbf{Z}^{\star},\mathbf{\Lambda}^{\star})\\
		\geq&0,%
	\end{align*}%
	where the last inequality holds due to the saddle-point interpretation \cite[Chapter 5.4.2]{Boyd2004}.
	Then, 
	letting
	\(
	\mathbf\Theta^\star
	=
	\operatorname{col}
	\{\mathbf U^\star,\mathbf Z^\star,\mathbf\Lambda^\star\}
	\),
	it follows from (\ref{M1}) that
\begin{align}
	&\frac{1}{2}\Big[\left\|\mathbf{\Theta}^{\star}-\mathbf{\Theta}^{q}\right\|_{\mathcal{G}_{1}}^{2}-\left\|\mathbf{\Theta}^{\star}-\mathbf{\Theta}^{q+1}\right\|_{\mathcal{G}_{1}}^{2}-\left\|\mathbf{\Theta}^{q}-\mathbf{\Theta}^{q+1}\right\|_{\mathcal{G}_{2}}^{2}\Big] \nonumber\\
	\geq &L(\mathbf{U}^{q+1},\mathbf{Z}^{q+1},\mathbf{\Lambda}^{\star})-L(\mathbf{U}^{\star},\mathbf{Z}^{\star},\mathbf{\Lambda}^{\star})
	\geq 0,\label{M2}
\end{align}
which implies
\begin{align}
	\begin{aligned}
		&\frac{1}{2}\left\|\mathbf{\Theta}^{\star}-\mathbf{\Theta}^{q}\right\|_{\mathcal{G}_{1}}^{2}-\frac{1}{2}\left\|\mathbf{\Theta}^{\star}-\mathbf{\Theta}^{q+1}\right\|_{\mathcal{G}_{1}}^{2}  
		\\
		\geq & \frac{1}{2}\left\|\mathbf{\Theta}^{q}-\mathbf{\Theta}^{q+1}\right\|_{\mathcal{G}_{2}}^{2} .
	\end{aligned}\label{M3}
\end{align}
According to condition~\eqref{conver}, we have
\(
G_U^{-1}-L_\delta^U>\mathbf 0.
\)
Moreover, letting
$D_{\mathcal L}
=\operatorname{diag}\{\mathcal L_{11},\ldots,\mathcal L_{NN}\}$,
for the undirected communication graph it holds that
\(
2D_{\mathcal L}-\mathcal L\geq\mathbf 0,
\)
which yields
$
	G_Z^{-1}-L_\delta^Z-\rho\tilde{\mathcal L}>\mathbf 0$ and
$
G_Z^{-1}-\rho\tilde{\mathcal L}
>\mathbf 0.
$
Since $\tilde{\mathcal L}\geq\mathbf 0$, it follows that
$
\mathcal G_1\geq\mathbf 0$,
\(\mathcal G_2\geq\mathbf 0,
\)
with the first two diagonal blocks of both matrices being
positive definite.
Therefore, \eqref{M3} yields that $\{\|\mathbf{\Theta}^{\star}-\mathbf{\Theta}^{q}\|_{\mathcal{G}_{1}}$, $q=0,1,\ldots \}$ is non-increasing and the primal
sequences $\{\mathbf{U}^q\}$ and $\{\mathbf{Z}^q\}$ are bounded.
Summing~(\ref{M2}) from $q=0$ to $q=t-1$ and dropping the
nonnegative terms yields $\sum_{q=0}^{t-1}
\left[
L(\mathbf U^{q+1},\mathbf Z^{q+1},\mathbf \Lambda^\star)
- L(\mathbf U^\star,\mathbf Z^\star,\mathbf\Lambda^\star)
\right]
\le
\frac12
\|\mathbf\Theta^\star-\mathbf\Theta^0\|_{\mathcal G_1}^2$.
Since each term in the summation is nonnegative, it follows that
\begin{align}
\lim_{q \rightarrow+\infty}\big[L(\mathbf U^{q},\mathbf Z^{q},\mathbf \Lambda^\star)
- L(\mathbf U^\star,\mathbf Z^\star,\mathbf\Lambda^\star)\big]=0.\label{L1}
\end{align}
Since $Q_i>\mathbf 0, R_i>\mathbf 0,P_i>\mathbf{0}$, 
the quadratic cost $\mathbf J(\mathbf U,\mathbf Z)$ is
strongly convex with respect to $(\mathbf U,\mathbf Z)$.
Moreover, the term
$\mathbf\Lambda^{\star\top}\tilde{\mathcal L}\mathbf Z$
is linear in $\mathbf Z$. Hence,
$L(\mathbf U,\mathbf Z,\mathbf\Lambda^\star)$ is also strongly
convex with respect to $(\mathbf U,\mathbf Z)$. 
Therefore, there
exists a constant $\mu>0$ such that
\(
L(\mathbf U^{q},\mathbf Z^{q},\mathbf \Lambda^\star)
- L(\mathbf U^\star,\mathbf Z^\star,\mathbf\Lambda^\star)
\geq \frac{\mu}{2}(\|\mathbf U^{q}-\mathbf U^{\star}\|^2+\|\mathbf Z^{q}-\mathbf Z^{\star}\|^2)\).
Together with \eqref{L1}, it follows that the sequence $\{(\mathbf{U}^{q},\mathbf{Z}^{q})$, $q=1,2,\ldots\}$ generated by Algorithm~2 converges to the optimal solution $(\mathbf{U}^{\star},\mathbf{Z}^{\star})$,
which completes the proof.\qed

\section{Proof of Corollary~1}

Before the proof begins, two useful lemmas are given.

\begin{lemma}\label{pro3}
	The candidate solution $(\hat{\mathbf{U}},\hat{\mathbf{Z}},\hat{\mathbf{\Lambda}})$ defined in (\ref{can}) satisfies that 
	\begin{align*}
		&(\hat{\mathbf{Z}}-\mathbf{Z})^{\top}\tilde{\mathcal{L}}\hat{\mathbf{\Lambda}}-(\hat{\mathbf{\Lambda}}-\mathbf{\Lambda})^{\top}\tilde{\mathcal{L}}\hat{\mathbf{Z}}\\
		=&\frac{1}{t}\sum_{q=1}^{t}\big[(\mathbf{Z}^q-\mathbf{Z})^\top \tilde{\mathcal{L}}\mathbf{\Lambda}^q-(\mathbf{\Lambda}^q-\mathbf{\Lambda)}^{\top}\tilde{\mathcal{L}}\mathbf{Z}^q\big]
	\end{align*}
	for all $\mathbf{\Lambda}$ and $\mathbf{Z}\in\tilde{\mathcal{Z}}$.
\end{lemma}
\noindent
\textbf{Proof of Lemma \ref{pro3}.} It is verified that $\forall\mathbf{\Lambda}, \forall\mathbf{Z}\in\tilde{\mathcal{Z}}$,
\begin{align}
	\begin{aligned}
		&(\hat{\mathbf{Z}}-\mathbf{Z})^{\top}\tilde{\mathcal{L}}\hat{\mathbf{\Lambda}}-(\hat{\mathbf{\Lambda}}-\mathbf{\Lambda})^{\top}\tilde{\mathcal{L}}\hat{\mathbf{Z}}\\
		=&(\hat{\mathbf{Z}}-\mathbf{Z})^{\top}\tilde{\mathcal{L}}{\mathbf{\Lambda}}-(\hat{\mathbf{\Lambda}}-\mathbf{\Lambda})^{\top}\tilde{\mathcal{L}}{\mathbf{Z}}\\	
		=&\frac{1}{t}\sum_{q=1}^{t}\big[(\mathbf{Z}^q-\mathbf{Z})^\top \tilde{\mathcal{L}}\mathbf{\Lambda}-(\mathbf{\Lambda}^q-\mathbf{\Lambda)}^{\top}\tilde{\mathcal{L}}\mathbf{Z}\big]\\
		=&\frac{1}{t}\sum_{q=1}^{t}\big[(\mathbf{Z}^q-\mathbf{Z})^\top \tilde{\mathcal{L}}\mathbf{\Lambda}^q-(\mathbf{\Lambda}^q-\mathbf{\Lambda)}^{\top}\tilde{\mathcal{L}}\mathbf{Z}^q\big],
	\end{aligned}
\end{align}
thereby completing the proof. \qed

\begin{lemma}\label{lemma4}
	(\cite{Xu2018a})
	Suppose that $\bar{x}\in\mathbb{X}$ is an approximate solution of the convex optimization problem $f^{\star}\triangleq\inf\{f(x)\mid Ax-b=\mathbf{0},x\in\mathbb{X}\}$. If there exist $\gamma>0$ and $\epsilon>0$ satisfying 
	\begin{align*}
		f(\bar{x})-f^{\star}+\gamma\|A \bar{x}-b\| \leq \epsilon,\left\|\lambda^{\star}\right\| < \gamma,
	\end{align*}
	then, it follows
	\begin{align}
		f(\bar{x})-f^{\star} \leq \epsilon,\|A \bar{x}-b\| \leq \frac{\epsilon}{\gamma-\left\|\lambda^{\star}\right\|},
	\end{align}
	where {$\lambda^{{\star}}$} is an optimal Lagrangian multiplier associated with the constraint $Ax-b=\mathbf{0}$ in the problem $f^{\star}$.
\end{lemma}

\noindent
\textbf{Then, the proof of Corollary~1 is given as follows.}

\vspace{4pt}
\noindent
\textbf{Proof of Corollary~1.} First, $\forall\mathbf{\Lambda}$, $\forall\mathbf{U}\in\tilde{\mathcal{U}}$, $\forall\mathbf{Z}\in\tilde{\mathcal{Z}}$, we have
\begin{align}
	&J(\hat{\mathbf{U}},\hat{\mathbf{Z}})-J({\mathbf{U}},{\mathbf{Z}})+(\hat{\mathbf{Z}}-\mathbf{Z})^{\top}\tilde{\mathcal{L}}\hat{\mathbf{\Lambda}}-(\hat{\mathbf{\Lambda}}-\mathbf{\Lambda})^{\top}\tilde{\mathcal{L}}\hat{\mathbf{Z}}\nonumber\\
	&\leq \frac{1}{t}\sum_{q=1}^{t}\big[J({\mathbf{U}^q},{\mathbf{Z}^q})-J({\mathbf{U}},{\mathbf{Z}})+(\mathbf{Z}^q-\mathbf{Z})^\top \tilde{\mathcal{L}}\mathbf{\Lambda}^q\nonumber\\
	&~~~~~~~~~~~~~~~~~~~~~~~~~~~~~~~~~~~~~~-(\mathbf{\Lambda}^q-\mathbf{\Lambda})^{\top}\tilde{\mathcal{L}}\mathbf{Z}^q\big],\nonumber\\
	&\leq \frac{1}{2t}\sum_{q=1}^{t}\big(\left\|\mathbf{\Theta}-\mathbf{\Theta}^{q-1}\right\|_{\mathcal{G}_{1}}^{2}-\left\|\mathbf{\Theta}-\mathbf{\Theta}^{q}\right\|_{\mathcal{G}_{1}}^{2}  \big)\\
	&\leq \frac{1}{2t}\big(\left\|\mathbf{U}-\mathbf{U}^{0}\right\|_{G_U^{-1}}^{2}+\left\|\mathbf{Z}-\mathbf{Z}^{0}\right\|_{G_Z^{-1}-\rho\tilde{\mathcal{L}}}^{2}+\left\|\mathbf{\Lambda}-\mathbf{\Lambda}^{0}\right\|_{\frac{\tilde{\mathcal{L}}}{\rho}}^{2}   \big),\nonumber
\end{align}
where the first inequality is obtained due to the convexity of $J({\mathbf{U}},{\mathbf{Z}})$ and Lemma 3, and the second inequality is valid owing to (\ref{M1}).

Then, define the set $\mathcal{B}_\gamma =\{\mathbf{\Lambda}|\|\mathbf{\Lambda}\|\leq \gamma\}$. For the optimal solution $({\mathbf{U}^{\star}},{\mathbf{Z}^{\star}})$, it follows that
{\small\begin{align}
	&\sup_{{\mathbf{\Lambda}}\in \mathcal{B}_\gamma}\Big\{J(\hat{\mathbf{U}},\hat{\mathbf{Z}})-J({\mathbf{U}^{\star}},{\mathbf{Z}^{\star}})+(\hat{\mathbf{Z}}-\mathbf{Z}^{\star})^{\top}\tilde{\mathcal{L}}\hat{\mathbf{\Lambda}}-(\hat{\mathbf{\Lambda}}-\mathbf{\Lambda})^{\top}\tilde{\mathcal{L}}\hat{\mathbf{Z}}\Big\}\nonumber\\
	=&\sup_{{\mathbf{\Lambda}}\in \mathcal{B}_\gamma}\big\{J(\hat{\mathbf{U}},\hat{\mathbf{Z}})-J({\mathbf{U}^{\star}},{\mathbf{Z}^{\star}})+\mathbf{\Lambda}^{{\top}}\tilde{\mathcal{L}}\hat{\mathbf{Z}}\big\}\nonumber\\
=&J(\hat{\mathbf{U}},\hat{\mathbf{Z}})-J({\mathbf{U}^{\star}},{\mathbf{Z}^{\star}})+\gamma\|\tilde{\mathcal{L}}\hat{\mathbf{Z}}\|\nonumber\\
	\leq&\sup_{{\mathbf{\Lambda}}\in \mathcal{B}_\gamma}\frac{1}{2t}\Big(\|\mathbf{U}^{\star}-\mathbf{U}^{0}\|_{G_U^{-1}}^{2}+\|\mathbf{Z}^{\star}-\mathbf{Z}^{0}\|_{G_Z^{-1}-\rho\tilde{\mathcal{L}}}^{2}+\|\mathbf{\Lambda}-\mathbf{\Lambda}^{0}\|_{\frac{\tilde{\mathcal{L}}}{\rho}}^{2}  \Big)\nonumber\\
	\leq&\frac{1}{2t}\Big[\|\mathbf{U}^{\star}-\mathbf{U}^{0}\|_{G_U^{-1}}^{2}+\|\mathbf{Z}^{\star}-\mathbf{Z}^{0}\|_{G_Z^{-1}-\rho\tilde{\mathcal{L}}}^{2}+\frac{\gamma^2}{\rho}\lambda_{\max }(\mathcal{L}) \Big].\label{eq:cor_combined}
\end{align}}%
Since $\gamma>\|\mathbf{\Lambda}^\star\|$, applying
Lemma~4 to \eqref{eq:cor_combined} yields the desired conclusion \eqref{rate}.
\qed

\section{Proof of Theorem 2}
\noindent
\textbf{Proof of Theorem 2.} 
Suppose that there exists an optimal solution $\mathbf{u}_i^{\star}(t_k)=\operatorname{col}\big\{u_i^{\star}(0|t_k),u_i^{\star}(1|t_k),\ldots,u_i^{\star}(\mathcal{T}-1|t_k)\big\}$ and $z_i^{\star}(t_k)$ at time $t_k$, $k\geq 0$. Subsequently, the corresponding optimal state $x_i^{\star}(l|t_k)$, $l\in[0,\mathcal{T}]$ of each agent can be obtained according to the system dynamics \eqref{dyn}.
Consider the candidate sequences at the next prediction time $t_{k+1}$ as $\bar{z}_i(t_{k+1})=A_i{z}^{\star}_i(t_k)+B_iD_i{z}^{\star}_i(t_k)={z}^{\star}_i(t_k)$,
{\small\begin{align}
		\bar{u}_i(l|t_{k+1})&=\begin{cases}
			u_i^{\star}(l+\Delta t|t_{k}), ~~~~~~~~~~~~~~~~~~l\in[0,\mathcal{T}-\Delta t), \\
			K_i\big(\bar{x}_i(l|t_{k+1})-\bar{z}_i(t_{k+1})\big)+D_i\bar{z}_i(t_{k+1}),  \\
			~~~~~~~~~~~~~~~~~~~~~~~~~~~~~~l\in[\mathcal{T}-\Delta t,\mathcal{T}-1],
		\end{cases}\label{Ufea}\\
		\bar{x}_i(l|t_{k+1})&=\begin{cases}
			x_i^{\star}(l+\Delta t|t_{k}),  ~~~~~~~~~~~~~~~~~~l\in[0,\mathcal{T}-\Delta t], \\
			(A_i+B_iK_i)^{l-\mathcal{T}+\Delta t}\big(x_i^{\star}(\mathcal{T}|t_k)-z_i^{\star}(t_k)\big)+z_i^{\star}(t_k),\\ ~~~~~~~~~~~~~~~~~~~~~~\qquad\quad
			l\in[\mathcal{T}-\Delta t+1,\mathcal{T}].
		\end{cases}\label{Xfea}
\end{align}}%
It follows from the first condition in \eqref{XTi} that
	$	\bar u_i(l|t_{k+1})\in{\mathcal U}_i$, $\forall l\in[0,\mathcal{T}-1]$.
	According to the second condition in \eqref{XTi}, the corresponding system state satisfies
	$	\bar x_i(l|t_{k+1})\in{\mathcal X}_i$, $\forall l\in[0,\mathcal{T}-1]$, and $	\bar x_i(\mathcal{T}|t_{k+1})\in{\mathcal X}_{i\mathcal{T}}$.
	Therefore, 
	it is verified that \eqref{XTi} guarantees that (\ref{Ufea})-(\ref{Xfea}) are feasible solutions of (\ref{problem1}) at $t_{k+1}$, thereby demonstrating the recursive feasibility of the proposed MPC framework.

For the feasible candidate solution \eqref{Ufea}-\eqref{Xfea}, we define the error variables $\bar{e}_{xi}(l |t_{k+1})=\bar{x}_i(l|t_{k+1})-\bar{z}_i(t_{k+1})$ for $l\in[0,\mathcal{T}]$ and $\bar{e}_{ui}(l |t_{k+1})=\bar{u}_i(l|t_{k+1})-D_i\bar{z}_i(t_{k+1})$ for $l\in[0,\mathcal{T}-1]$. Similarly, given the optimal solution $\big(\mathbf{u}_i^{\star}(t_k),z_i^{\star}(t_k)\big)$, we define the error variables ${e}_{xi}^{\star}(l |t_{k})={x}_i^{\star}(l|t_{k})-{z}_i^{\star}(t_{k})$ for $l\in[0,\mathcal{T}]$ and ${e}_{ui}^{\star}(l |t_{k})={u}_i^{\star}(l|t_{k})-D_i{z}^{\star}_i(t_{k})$ for $l\in[0,\mathcal{T}-1]$.
Then, for the optimal solution $\big(\mathbf{U}^{\star}(t_{k+1}),{\mathbf{Z}}^{\star}(t_{k+1})\big)$ at time $t_{k+1}$, we have
\begin{align}
		&\mathbf{J}(\mathbf{U}^{\star}(t_{k+1}),\mathbf{Z}^{\star}(t_{k+1}))-\mathbf{J}(\mathbf{U}^{\star}(t_{k}),\mathbf{Z}^{\star}(t_{k}))\nonumber\\
		\leq& \mathbf{J}(\bar{\mathbf{U}}(t_{k+1}),\bar{\mathbf{Z}}(t_{k+1}))-\mathbf{J}(\mathbf{U}^{\star}(t_{k}),{\mathbf{Z}}^{\star}(t_{k}))\nonumber\\
		=&\sum_{i=1}^{N}\bigg\{\sum_{l=0}^{\mathcal{T}-1}\Big[\|\bar{e}_{xi}(l| t_{k+1})\|_{Q_{i}}^{2}+\|\bar{e}_{ui}(l |t_{k+1})\|_{R_{i}}^{2}\Big]\nonumber\\
		&~~~~~~~~~~~~+\|\bar{e}_{xi}(\mathcal{T} | t_{k+1})\|_{P_{i}}^{2}-\|e_{xi}^{\star}(\mathcal{T} |t_{k})\|_{P_{i}}^{2}\nonumber\\
		&~~~~~~-
		\sum_{l=0}^{\mathcal{T}-1}\Big[\|e_{xi}^{\star}(l| t_{k})\|_{Q_{i}}^{2}-\|e_{ui}^{\star}(l| t_{k})\|_{R_{i}}^{2}\Big]\bigg\}\nonumber\\
		= & \sum_{i=1}^{N} \bigg\{-\sum_{l=0}^{\Delta t -1}\Big[
		\|e_{xi}^{\star}(l| t_{k})\|_{Q_{i}}^{2}+\|e_{ui}^{\star}(l| t_{k})\|_{R_{i}}^{2}\Big]\nonumber\\
		&~~~~~~~~+\|{e}^{\star}_{xi}(\mathcal{T} | t_{k})\|_{	\Psi_i}^2\bigg\}\nonumber\\
		\leq&{-\sum_{i=1}^{N}\sum_{l=0}^{\Delta t -1}\Big[
			\|e_{xi}^{\star}(l| t_{k})\|_{Q_{i}}^{2}+\|e_{ui}^{\star}(l| t_{k})\|_{R_{i}}^{2}\Big]},\label{J}
\end{align}%
where
$\Psi_i=\sum_{h=0}^{\Delta t -1}\|(A_i+B_iK_i)^{h}\|^2_{Q_i+K_i^{\top}R_iK_i} + \|(A_i+B_iK_i)^{\Delta t}\|_{P_i}^2-P_i$, the second equality holds by substituting (\ref{Ufea})-(\ref{Xfea}), and the penultimate inequality holds due to (\ref{Ki}).

According to~(\ref{J}), the sequence
$\{\mathbf J(\mathbf U^\star(t_k),\mathbf Z^\star(t_k)),
k=0,1,\ldots\}$ is non-increasing since $Q_i>\mathbf{0}$ and
$R_i>\mathbf{0}$.
Moreover, considering
$\mathbf J(\mathbf U^\star(t_k),\mathbf Z^\star(t_k))\geq 0$,
it is lower bounded. Then,
summing~(\ref{J}) from $k=0$ to $k=K$ gives
\begin{align}
	&\sum_{k=0}^{K}\sum_{i=1}^{N}\sum_{l=0}^{\Delta t-1}
	\left[
	\|e_{x_i}^\star(l|t_k)\|_{Q_i}^{2}
	+
	\|e_{u_i}^\star(l|t_k)\|_{R_i}^{2}
	\right]
	\nonumber\\
	\quad&\leq
	\mathbf J(\mathbf U^\star(t_0),\mathbf Z^\star(t_0))
	-
	\mathbf J(\mathbf U^\star(t_{K+1}),\mathbf Z^\star(t_{K+1}))
	\nonumber\\
	\quad&\leq
	\mathbf J(\mathbf U^\star(t_0),\mathbf Z^\star(t_0)).
	\label{summable}
\end{align}
Letting $K\rightarrow\infty$, it follows that
\begin{align}
	\sum_{k=0}^{\infty}\sum_{i=1}^{N}\sum_{l=0}^{\Delta t-1}
	\left(
	\|e_{x_i}^\star(l|t_k)\|_{Q_i}^{2}
	+
	\|e_{u_i}^\star(l|t_k)\|_{R_i}^{2}
	\right)<\infty.
\end{align}
Since all terms in the above summation are nonnegative and
$Q_i>\mathbf{0}$, $R_i>\mathbf{0}$, it follows that
\begin{align}
	\lim_{k\rightarrow\infty}
	e_{x_i}^\star(l|t_k)=\mathbf{0},\qquad
	\lim_{k\rightarrow\infty}
	e_{u_i}^\star(l|t_k)=\mathbf{0} \label{exconv}
\end{align}
for all
$i\in\{1,2,\ldots,N\}$ and $l\in[0,\Delta t-1]$. 

Recall that 
at each prediction time $t_k$, there have
\(
e_{x_i}^\star(l|t_k)
=
x_i^\star(l|t_k)-z_i^\star(t_k)
\) and
\(
z_i^\star(t_k)=z_j^\star(t_k)\),
\(\forall i,j\in\{1,2,\ldots,N\}
\).
Therefore, for any $i,j\in\{1,2,\ldots,N\}$ and
$l\in[0,\Delta t-1]$, we have $
	\|x_i^\star(l|t_k)-x_j^\star(l|t_k)\|
\leq	
	\|x_i^\star(l|t_k)-z_i^\star(t_k)\|
	+
	\|z_i^\star(t_k)-z_j^\star(t_k)\|+
	\|z_j^\star(t_k)-x_j^\star(l|t_k)\|=
	\|e_{x_i}^\star(l|t_k)\|
	+
	\|e_{x_j}^\star(l|t_k)\|.
$
Together with~(\ref{exconv}), this yields
\begin{align}
	\lim_{k\rightarrow\infty}
	\|x_i^\star(l|t_k)-x_j^\star(l|t_k)\|
	=0,
\end{align}
for all $i,j\in\{1,2,\ldots,N\}$ and
$l\in[0,\Delta t-1]$.
Then, under the MPC framework, ${x}_i^{\star}(l|t_{k})$ for $l\in[0,\Delta t-1]$ is the real-time system state over the interval
$[t_k,t_{k+1}-1]$. Hence, we have $\lim_{t\rightarrow \infty}\big( x_i(t)-x_j({t})\big)=\mathbf{0}$, $\forall i,j\in\{1,2,\ldots,N\}$, indicating that the closed-loop multi-agent system asymptotically achieves consensus.
Together with the recursive feasibility
established above, all state and input constraints are satisfied
for all time, which completes the proof.
\qed

\section*{References}
\vspace{-1.5em}
\bibliographystyle{ieeetr}
\bibliography{lq.bib}

\end{document}